\documentclass{article}

\usepackage{PRIMEarxiv}

\usepackage[utf8]{inputenc} 
\usepackage[T1]{fontenc}    
\usepackage{hyperref}       
\usepackage{url}            
\usepackage{booktabs}       
\usepackage{amsfonts}       
\usepackage{nicefrac}       
\usepackage{microtype}      
\usepackage{lipsum}
\usepackage{fancyhdr}       
\usepackage{graphicx}       
\graphicspath{{media/}}     
\usepackage{amsmath,amsfonts}
\usepackage{algorithmic}
\usepackage{algorithm}
\usepackage{array}
\usepackage{float}
\usepackage{subfig}
\usepackage{cite}
\usepackage{amssymb}
\usepackage{color}
\newtheorem{remark}{Remark}[section]

\newtheorem{thm}{Theorem}[section]
\newtheorem{corollary}{Corollary}[section]
\newtheorem{proof}{Proof}[section]

\title{Optimal Control for Unmanned Systems with One-way Broadcast Communication
}

\author{
  Chao Ge \\
  School of Mathematical Sciences \\
  University of Chinese Academy of Sciences \\
  Beijing, China\\
  \texttt{gechao@amss.ac.cn} \\
   \And
  Ge Chen \\
  Key Laboratory of Systems and Control \\
  Academy of Mathematics and Systems Science, Chinese Academy of Science \\
  Beijing, China\\
  \texttt{chenge@amss.ac.cn} \\
}

\begin{document}
\maketitle

\begin{abstract}
Unmanned systems (USs) including unmanned aerial vehicles, unmanned underwater vehicles, and unmanned ground vehicles have great application prospects in military and civil fields, among which the process of finding feasible and optimal paths for the agents in USs is a kernel problem.
Traditional path finding algorithms are hard to adequately obtain optimal paths in real-time under fast time-varying and poor communication environments.
We propose an online optimal control algorithm for USs based on a one-way broadcast communication mode under the assumption of a poor communication environment, mobile targets, radars (or sonar), and missiles
(or torpedoes).
With the principle of receding
horizon control, optimal (or suboptimal) paths are then generated by the approximation theory of neural networks and gradient optimization techniques, with low computation requirements.
Also, we give a convergence analysis for our algorithm, and show that each agent can reach its target in finite time  under some conditions on  agents, targets and radar-missiles. Moreover, simulations demonstrate that the agents in USs can generate optimal (or suboptimal) paths in real time using our algorithm while effectively avoiding collision with other agents or detection by enemy radars.
\end{abstract}


\section{Introduction}\label{intro}
Due to to the advantages of high efficiency/cost ratio and strong adaptability,
the application of unmanned systems (USs) including unmanned aerial vehicles (UAVs), unmanned underwater vehicles (UUVs), and unmanned ground vehicles (UGVs) for surveillance and striking against mobile targets is increasing.
Significant development has been achieved by military and civil fields so far.
It is essential to find feasible and optimal or suboptimal real-time paths for agents to enhance USs control further.
The paths must satisfy certain constraints, such as the maneuverability of agents and the communication environment. Additionally, the cost and antagonistic factors are also common factors that need to consider, which can contribute to the optimization of the path performance\cite{aggarwal2020path,bortoff2000path}.

Path finding is a traditional problem for USs and has been extensively researched in \cite{aggarwal2020path,lavalle2006planning,hsu2002randomized,wang2020optimal}.
Generally, the research can be divided into two categories according to the real-time abilities of algorithms.
Off-line path planning algorithms are based on complete prior information of the environment and require all the environment information to plan the desired paths before execution.
Among this category of research, the most commonly used off-line algorithms can be further subdivided into: i) search algorithms such as rapid-exploring random trees algorithm\cite{lavalle1998rapidly,kuffner2000rrt}, A* algorithm\cite{hart1968formal,masehian2007classic}, gravitational search algorithm\cite{rashedi2009gsa,li2012path}; ii) swarm intelligence method such as ant colony optimization\cite{ma2018path}, particle swarm optimization\cite{fu2011phase,wu2020cooperative}, genetic algorithm\cite{pehlivanoglu2012new}; iii) optimal control-based method\cite{bai2018integrated,bai2017clustering}, etc.
Off-line path-planning algorithms rely on pre-obtained environment information to a great extent. Hence, they are mainly used in a static environment or a dynamic environment with the known motion law.
Algorithms are designed to find paths in real-time by modifying off-line algorithms with the methods
including  dynamic window approach\cite{wang2020dynamics}, model predictive control\cite{liu2017path}, dynamic neural network model\cite{chen2020optimal}, deep neural network model\cite{wang2021real}, and parallel computing\cite{roberge2012comparison}. However, such methods require complex computation and could be time-consuming. Besides, they may experience problems of local optima and adaptability. The potential field method\cite{warren1989global,ge2002dynamic,chen2016uav} is also suitable for real-time applications. This method considers that the agents in USs move according to a combined force of the attractive field and repulsive potential of targets and obstacles. Although, this is a simple method to obtain feasible paths for USs, it does not consider  path optimization.
In addition, a poor communication environment can further complicate path finding which brings real-time algorithms a great deal of difficulty and challenge. The agents in USs may encounter the following difficulties in real-time path finding:
\begin{itemize}
\item Limited communication channel capacity, which inhibits one-to-one communication, especially in the case of large clusters.
\item Enemy radio suppression. Due to the size, cost, and energy constraints of USs' agents, it is difficult for  agents to send signals and these signals can easily be suppressed. Thus, agents are susceptible to interference.
\item Some scenes will require radio silence to enhance concealment, such as raid operations.
\item  When the application scenario moves underwater, two-way communication becomes hard or risky. Seawater has a strong absorption effect on electromagnetic wave energy. The shorter the wavelength of electromagnetic waves, the greater the attenuation in seawater. Therefore, short-wave attenuation in water is rapid, and communication mainly uses long-wave.
However, the UUVs such as submarines only receive one-way long waves from the outside world because the transmitter is too large.
When necessary to transmit information, the submarines float or use floating antennas for shortwave communication which is at risk of being detected.
\end{itemize}

To better adapt to such poor communication environment and satisfy the real-time requirement,
we propose a new one-way broadcast communication framework.
In this framework, we combines and modifies the weighted Hungarian method, neural network algorithm, and gradient optimization techniques, and
present an online algorithm for optimal path control of USs in this paper.
The one-way broadcast communication refers to that each agent of USs only receives signals from the command center but does not send any signals to the command center or any other agents.
In detail, the command center detects the positions of enemy radars and missiles, estimates the positions of agents, allocates targets for agents by the weighted Hungarian method, and processes such information into broadcast signals at set intervals.
After receiving signals from the command center, by the receding
horizon control way, agents compute optimal paths using the approximation theory of neural networks and gradient optimization techniques. Using this method, agents can cooperatively reach several mobile targets and avoid collision and detection in a dynamic and antagonistic environment with the least cost.

\emph{Contributions:}
First, to enhance practicality, we consider a dynamic and antagonistic environment, which includes poor and limited communication, mobile targets, enemy radars (or sonar), and missiles. These factors make the construction and optimization of USs difficult, and traditional off-line  or real-time algorithms are hard to apply.
Different from previous work\cite{liu2017path,chen2020optimal,wang2021real}, our algorithm introduces a one-way communication framework in which
the agents receive broadcast information from their command center.
This framework can reduce the communication requirements of agents and should be more suitable for harsh or antagonistic environments.

Second, different from traditional potential field method\cite{warren1989global,ge2002dynamic,chen2016uav}, our algorithm considers the on-line path optimization for all agents, which may reduce their energy consumption. We adopt the simple but efficient gradient descent method to solve the optimization problem,
which has the advantage of low computation for each agent.

Third, we give some convergence analysis for our algorithm. It is shown that, under some conditions on  agents, targets and radar-missiles, each agent can reach its target in finite time. Meanwhile, the upper bound of arrival time of each agent is estimated.

Finally, we present some simulations to show the performance of our algorithm. Simulations show that our algorithm combined of neural networks and gradient optimization techniques
  is more sensitive than the basic gradient method. Meanwhile, each agent can reach its target in the dynamic environment successfully.
Moreover, the computational time in simulations indicates the low
computational requirement of our algorithm.

\emph{Organization}: In section \ref{description}, we first present some necessary assumptions and describe the optimal path control problem in detail. The algorithm is then presented and discussed in section \ref{algorithm}.
In section \ref{analysis}, we give our theoretic convergence analysis of our algorithm. Then, we present simulations to test the capabilities of our algorithm in section \ref{Simulations}. Finally, some conclusions are given in section \ref{Conclusions}.

\section{Problem Description}\label{description}
The agents in USs encounter numerous challenges in poor communication environments, as discussed in Section \ref{intro}.
We propose the one-way broadcast communication mode to address these difficulties in this work.
A dynamic and antagonistic environment is first considered, where antagonistic factors include  radars  and missiles  that are fixed on the ground or mounted on vehicles (e.g., Tor anti-aircraft missiles), or sonar and torpedoes placed on  aircraft or warships. To simplify the exposition,
 such antagonistic factors we  refer to as radar-missiles.
We consider the use of a command aircraft or a command center for broadcast in poor communication environments where two-way or one-to-one communication is impossible. The command aircraft may be equipped with an airborne moving target detection (MTD) radar, which can detect the positions of moving targets and enemy radar-missiles. In a similar manner, the command aircraft can detect the positions of agents or calculate and estimate their positions according to their initial positions and the movements of each step. These functions can also be carried out by a reconnaissance aircraft and communicated to the command center.
The command aircraft or command center measures costs and assigns targets to agents using the position information. This information is then processed into signals and sent to agents as one-way broadcasts. The agents compute optimal and safe paths according to the signals received, considering the costs, detection risk of radars, and the collision risk.

This one-way broadcast communication mode could also be applied to the next generation of fighter jets. One of the core technologies of the next generation fighter is artificial intelligence drone control technology.
Designing command aircraft for stealth comes at the expense of bomb load capacity. To make up for this imperfection, stealth command aircraft instructs payload UAVs to complete missions by broadcast.
This mode cannot easily be suppressed by radio and UAVs only need to receive signals, which reduces the cost of UAVs. The application scenario is illustrated in Figure \ref{Fig:fighter}. Unmanned submarines can also adopt this communication mode, where reconnaissance planes are used to detect the targets' positions and enemy sonar towed by anti-submarine aircraft. A shore-based radio transmitter then receives the information from the reconnaissance planes and broadcasts signals to the biomimetic fish-like submarines, which track the targets in a self-organizing way. This application scenario is illustrated in Figure \ref{Fig:submarine}.
\begin{figure}[!t]
\centering
\includegraphics[width=3.5in]{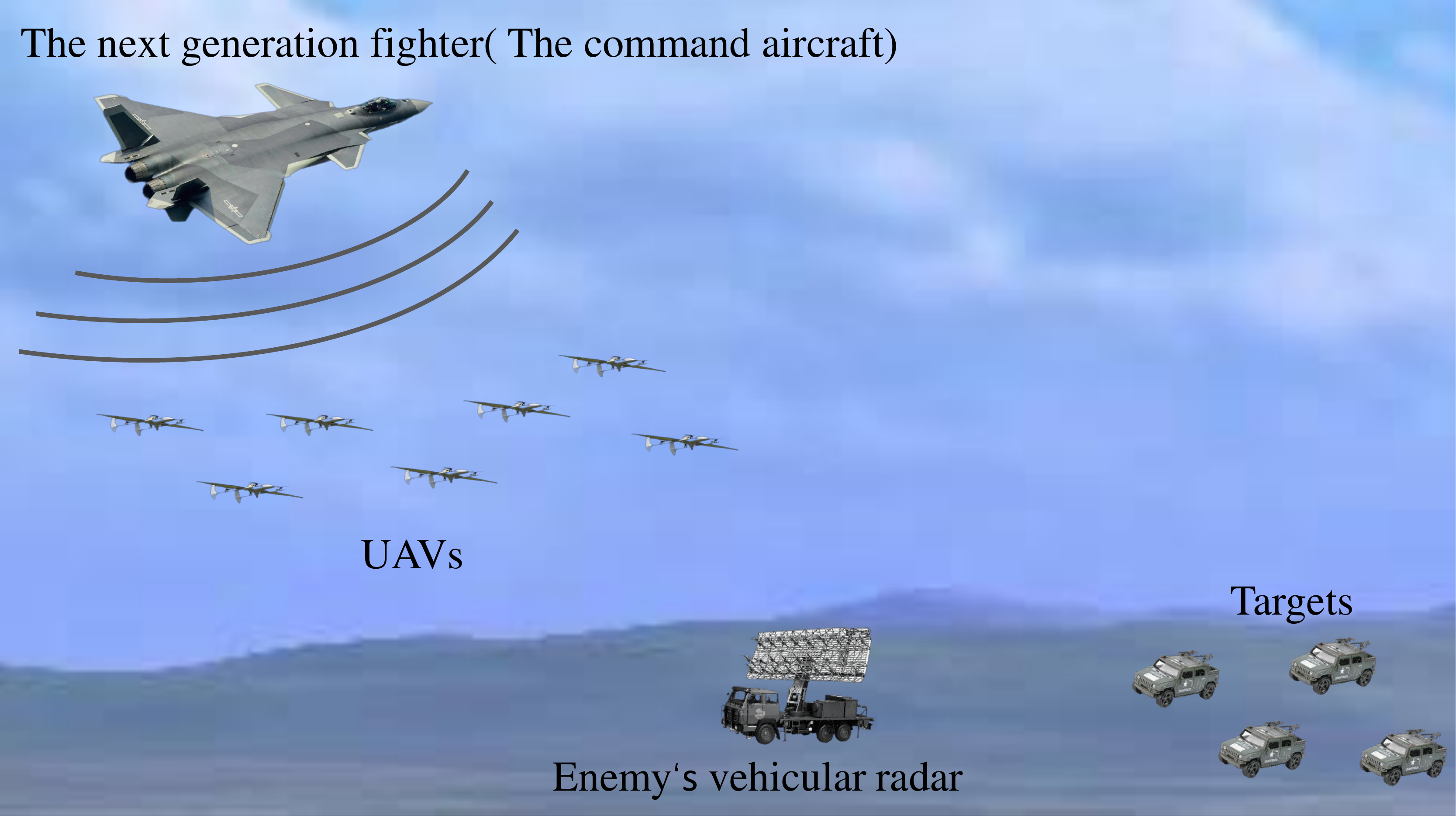}
\caption{One-way broadcast communication mode applied in the next generation of fighter jets.}
\label{Fig:fighter}
\end{figure}
\begin{figure}[!t]
\centering
\includegraphics[width=3.5in]{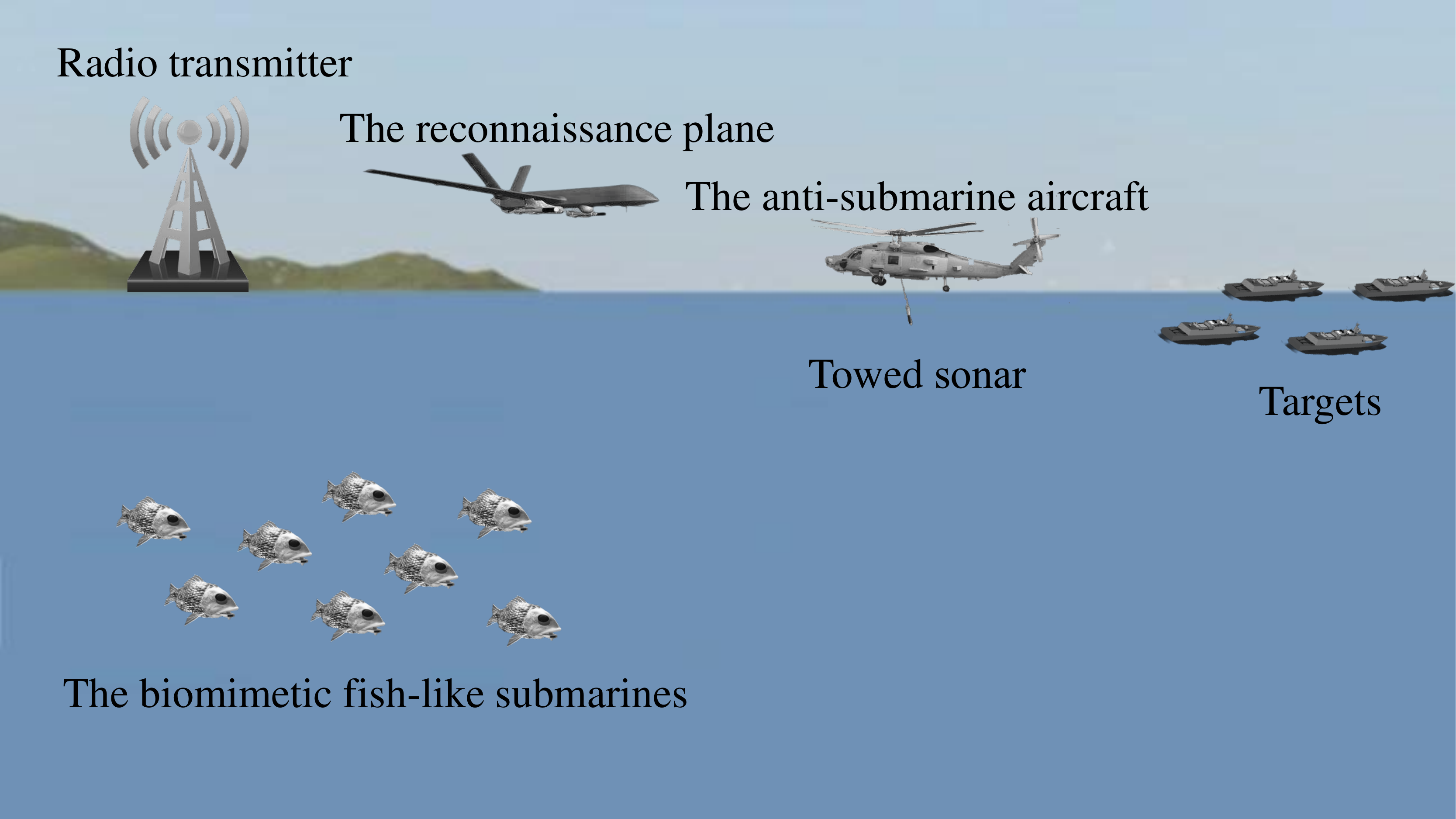}
\caption{One-way broadcast communication mode applied in unmanned submarines.}
\label{Fig:submarine}
\end{figure}

We propose the optimal path control problem of USs by constructing a model first.
To simplify the expression, we consider a command center, $N$ agents, $M$ enemy radar-missiles, and $K$ targets distributed in a two-dimensional bounded region $S$.
As it is assumed that one missile must be accompanied by one radar, they are put together, and the detection range of the radar is larger than the attack range of the missile.
The targets move bounded by a maximum velocity. We can only detect their current positions, but cannot predict their positions at the next moment. Also,  we only consider a single task whose
objective is that each target is reached by at least one agent with the shortest safe path. Therefore, the number of agents needs to be larger than or at least equal to the number of targets. Let:
\begin{itemize}

\item[(1)] $\{U_1,U_2,\ldots,U_N\}$ is a team of $N$ agents in USs. Let $\left(x_i^U(t),y_i^U(t)\right)\in\mathbb{R}^2$ denote the position of $U_i$ at time $t$.
\item[(2)] $\{O_1,O_2,\ldots,O_M\}$ is a set of $M$ mobile radar-missiles. Let $\left(x_j^O(t),y_j^O(t)\right)\in\mathbb{R}^2$ be the position of $O_j$ at time $t$. The detection range of a radar and attack range of a missile are $R_d$ and $R_a$ respectively. We assume $R_d>R_a>0$.
\item[(3)] $\{T_1,T_2,\ldots,T_K\}$ is a set of $K$ mobile targets. Let $\left(x_i^T(t),y_i^T(t)\right)\in\mathbb{R}^2$ be the position of the target $T_i$ at time $t$. We assume $K \leq N$.
\end{itemize}

We consider the following two cases to optimize agents' paths online  in all possible situations:\\
Case I: The command center knows the positions of the agents it commands. For example, in Figure \ref{Fig:fighter}, the command aircraft may detect the positions of agents it commands by radar.\\
Case II: The command center does not know the positions of the agents it commands. For example, in Figure \ref{Fig:submarine}, the command center cannot know the positions of the biomimetic fish-like submarines.

To be more practical, this paper also considers many factors, such as collision avoidance, antagonistic radars and missiles, the cost of paths, the maneuverability of the agents, and the cooperation scheme. These factors can be reflected by constructing a suitable objective function and some constraints, and then we can formulate the optimal path control problem.
First, we construct a cost function which is composed of the following three parts:
\begin{itemize}
\item[(1)] \textit{Total distance:}
	
\quad Because each agent can only carry limited energy, its moving range is limited. In order to complete the task better,  we need to minimize the total distance between each agent and its target\cite{aggarwal2020path}. The assigned target for the agent $U_i$ at time $t$ is $T_{\rm{tar},i}(t)\in \mathcal{T}$, whose position is denoted by  $\left(x_{\rm{tar},i}^T(t),y_{\rm{tar},i}^T(t)\right)$.
The distance from the agent $U_i$ to its assigned target $T_{\rm{tar},i}(t)$ at time $t$ is
\begin{equation}\label{fun1}
D_i^{U-T}(t)
:= \sqrt{\left(x_i^U(t) - x_{\rm{tar},i}^T(t)\right)^2 + \left(y_i^U(t) - y_{\rm{tar},i}^T(t) \right)^2 }.
\end{equation}

To simplify calculations, we let
\begin{equation}\label{fun1}
f_{i,1}(t):= \frac{1}{2} \left(D_i^{U-T}(t)\right)^2.
\end{equation}

\item[(2)] \textit{Radar-missile risk:}
	
\quad Since the antagonistic factors include the detection of radars and the attack of missiles, agents need to keep away from the region within the range of radar and missiles. Inspired by\cite{barraquand1994penalty}, we use a penalty function to represent the risk of radar detection and missile attack. Let

\begin{equation*}
D_{i,j}^{U-O}(t) := \sqrt{( (x_i^U(t) - x_{j}^O(t) ) )^2 + (y_i^U(t) - y_{j}^O(t) )^2 }
\end{equation*}
denote the distance between the agent $U_i$ and the  radar-missile $O_j$ at time $t$, and set
\begin{flalign*}
	&
	\alpha_{i,j}(t):= \left\{
	\begin{aligned}
		0\quad,  &\quad \text{if}\quad D_{i,j}^{U-O}(t)  > R_d\\
		1\quad,  & \quad \text{otherwise}
	\end{aligned}
	\right.,\\
	&
	\beta_{i,j}(t):= \left\{
	\begin{aligned}
		0\quad,  &\quad \text{if}\quad D_{i,j}^{U-O}(t)  > R_a \\
		1\quad,  & \quad \text{otherwise}
	\end{aligned}
	\right..
\end{flalign*}
The penalty function for the risk of radar detection and missile attack is formulated by
\begin{equation}\label{fun2}
	f_{i,2}(t):= k_d \sum_{j=1}^{M} \alpha_{i,j}(t)+k_a  \sum_{j=1}^{M} \beta_{i,j}(t),
\end{equation}
where  $k_a>k_d>0$ are constants denoting penalty factors for the possible radar detection and missile attack respectively.

\item[(3)] \textit{Collision risk:}

\quad In the process of agents' move, collision avoidance is also worthy of attention. The agents need to keep a safe distance from each other. Let $D_{\rm{safe}}>0$ be a constant denoting the safe distance between two agents, and
\begin{equation*}
D_{i,j}^{U-U}(t) := \sqrt{( (x_i^U(t) - x_j^U(t) ) )^2 + (y_i^U(t) - y_j^U(t) )^2 }
\end{equation*}
denotes the distance between two agents $U_i$ and $U_j$ at time $t$.
Set
\begin{flalign*}
	\gamma_{i,j}(t):= \left\{
	\begin{aligned}
	0\quad,  &\quad \text{if}\quad D_{i,j}^{U-U}(t)  >D_{\rm{safe}}\\
	1\quad,  & \quad \text{otherwise}
	\end{aligned}
	\right..
\end{flalign*}
  Similar to (\ref{fun2}), we use the following penalty function to represent the risk of collision:
\begin{equation}\label{fun3}
	f_{i,3}(t) := k_c\sum_{j=1}^{N} \gamma_{i,j}(t),
\end{equation}
where  $k_c>0$ is a constant denoting the penalty factor for the possible collision.
\end{itemize}

Next, we formulate constraint conditions for our optimization problem.
We assume that each agent is equipped with control apparatus of speed and heading angle, whose movement can be simplified into the following model:
\begin{flalign*}
	\left\{
	\begin{aligned}
		&\dot{x}_i^U(t) = v_i^U(t) \rm{cos}\psi_i(t)\\
		&\dot{y}_i^U(t) = v_i^U(t) \rm{sin}\psi_i(t)
	\end{aligned}
	\right.,
\end{flalign*}
where $v_i^U(t)= \sqrt{ \dot{x}_i^U(t) ^2+ \dot{y}_i^U(t)^2} \geq 0$ is the velocity of agent $U_i$ at time $t$, and $\psi_i(t)$ is the heading angle of agent $U_i$ at time $t$.
Let $L_i(t)$ be the total moving range of $U_i$ until time $t$. We assume the initial time of our algorithm is $0$ and then
	$$L_i(t) = \int_{0}^{t} v_i^U(s) ds .$$
We consider three restrictions on the agents' maneuverability and set a maximum velocity $v_{\rm{max}} > 0$, a maximum lateral overload $n_{\rm{max}} > 0$, and a maximum moving range $\bar{L} > 0$ for each agent.
Similar to \cite{ren2007constrained}, we construct the following constraint conditions: for any $t\geq 0$ and $1\leq i\leq N$,
\begin{flalign}\label{constraint}
	\left\{
	\begin{aligned}
		&  v_i^U(t)  \leq v_{\rm{max}}\\
		& |\dot{\psi}_i(t)| \leq \frac{n_{\rm{max}}\emph{g}}{v_i^U(t)}\\
		& L_i(t)  \leq \bar{L}
	\end{aligned}
	\right.,
\end{flalign}
where $\emph{g}\approx 9.8~\rm{m/s^2}$ is the gravitational acceleration.

For the first two constraints of (\ref{constraint}), we will set some limitations in the control of each agent. For the last constraint, we estimate the value $L_i(t)+D_i^{U-T}(t)$ which represents the current moving
range of agent $U_i$ plus the distance from its current position to its target. When this value exceeds
the maximum moving range $\bar{L}$,  we add a penalty item to the cost function. Let
\begin{equation}\label{fun4}
	f_{i,4}(t) := k_l \eta_{i}(t),
\end{equation}
where
$$
	\eta_{i}(t):= \left\{
	\begin{aligned}
		0\quad,  &\quad \text{if}\quad L_i(t)+D_i^{U-T}(t) \leq \bar{L}\\
		1\quad,  & \quad \text{otherwise}
	\end{aligned}
	\right.,
$$
and $k_l>0$ is a constant denoting
the penalty factor. Then, the final objective function is represented as:
\begin{equation}\label{newobjfun}
	H(t):= \sum_{i=1}^{N}\left[f_{i,1}(t) + f_{i,2}(t) + f_{i,3}(t)+  f_{i,4}(t)\right].
\end{equation}
By (\ref{constraint}) and (\ref{newobjfun}), we write the online optimal path control problem for agents as follows:
\begin{flalign}\label{opt_pro}
    &\min_{\{\left(x_i^U(t), y_i^U(t)\right)\}_{i=1}^N} H(t) \nonumber\\
	&~\mbox{s.t.}~~\left\{
	\begin{aligned}
		&  v_i^U(t)  \leq v_{\rm{max}}\\
		& |\dot{\psi}_i(t)| \leq \frac{n_{\rm{max}}\emph{g}}{v_i^U(t)}
	\end{aligned}
	\right., ~~i=1,\ldots,N.
\end{flalign}
Because $f_{i,2}(t)$, $f_{i,3}(t)$, and $f_{i,4}(t)$ contain several $0-1$ items, traditional gradient descent algorithms cannot be used directly to solve the problem (\ref{opt_pro}).
Let
\begin{equation}\label{fun5}
F(t) = \sum_{i=1}^{N}\left[f_{i,2}(t) + f_{i,3}(t)+  f_{i,4}(t)\right].
\end{equation}
Then, the non-differentiable part of $H(t)$ is $F(t)$.
We can use neural network algorithm to approximate $F(t)$ by a differentiable function $F^*(t)$, then use a gradient descent algorithm to optimize $F^*(t)+\sum_{i=1}^{N}f_{i,1}(t)$.
The detailed algorithm is proposed in the next section.

\section{Optimal control for the agents in USs with known or unknown position}\label{algorithm}
In this section,  we solve the Case I and Case II mentioned in section \ref{description}. This section is divided into four parts. We first use the weighted Hungarian method\cite{kuhn1955hungarian} to assign targets to the agents. This method is mainly applied to one-to-one match problems, and the number of targets equals the number of agents.
However, in our setting the number of agents can be bigger than the number of targets. To solve this contradiction, we increase the number of targets by copying them until their number equals the number of agents. That is, we create virtual duplication of the targets by duplicating their positions in turn until the total number of real and virtual targets equals the number of agents.
If one agent breaks down or is shot down, the total number of real and virtual targets will be reduced proportionally. Then, we delete one copied virtual target following a certain order and keep them equal in number. If two agents are assigned to the same target, both of them move towards the target.
In Subsection \ref{Subsec_Hun}, we assign targets to agents, and in Subsection \ref{Subsec_approx}, we find a differentiable function to approximate our objective function. In Subsection \ref{Subsec_plan}, we use gradient descent to solve our optimization problem and give our complete path finding algorithm.
In Subsection \ref{unknown}, we specifically elaborated the modified algorithm for Case II.
\subsection{The target assignment for the agents}\label{Subsec_Hun}
Noting that $f_{i,1}$ in the objective function needs to assign a target for each agent, we adopt the classic Hungarian method to solve the assignment problem.
The Hungarian method runs on a bipartite graph, which is a special kind of graph. If the vertex set $\mathcal{V}$ of Graph $\mathcal{G}=(\mathcal{V},\mathcal{E})$ can be divided into two disjoint sets $\mathcal{R}$ and $\mathcal{P}$, and no vertices in the same set are adjacent, the graph $\mathcal{G}$ is a bipartite graph. Let $\mathcal{R}=\mathcal{R}(t)$ denote the set of agents at time $t$, and $\mathcal{P}=\mathcal{P}(t)$ denote the set of targets at time $t$.  Because the Hungarian method only solves one-to-one match problems, we can increase the number of targets by copying them such that $|\mathcal{R}| = |\mathcal{P}|$.
Set $\mathcal{G}=\mathcal{G}(t)=(\mathcal{V},\mathcal{E},\omega)$ to be a weighted bipartite graph, where $\mathcal{V}=\mathcal{R}\cup \mathcal{P}$,  $\mathcal{E}$ consists of edges whose endpoints are an agent and a target,
and the weight to an edge is given by the distance between the corresponding agent and target at time $t$.
In detail, for any $i \in \mathcal{R}$ and $j \in \mathcal{P}$, let $$\omega(i,j)=\omega(i,j)(t) = \frac{\lambda}{D_{i,j}^{U-T}(t)},$$ where $\lambda$ is a scaling factor and $D_{i,j}^{U-T}(t)$ is the distance between agent $U_i$ and target $T_j$ at time $t$.
We use the weighted Hungarian method proposed in\cite{kuhn1955hungarian} to solve the assignment problem for the agents.
The weighted Hungarian method maximizes the sum of all edges' weight $\sum_i \omega(i,j)(t)$, that is, we minimize the total distance between agents and their assigned targets.
The details of this method are described in Appendix.

\begin{remark}
In the second section of\cite{chopra2017distributed}, the weighted Hungarian method is proved to converge to an optimal solution with a convergence rate $O(r^3)$, where $r = |\mathcal{R}| = |\mathcal{P}| = |\mathcal{V}|/2$.
\end{remark}

We can run the weighted Hungarian method in the command center to reassign targets for the agents at a timed period, or when an agent is broken or shot down.

\subsection{Approximation to objective function with a BP neural network}\label{Subsec_approx}
As stated in Section \ref{description}, we note that $F(t)$ is a step function that contains several $0-1$ items. Thus, the traditional gradient method may have a poor performance. We need to find a differentiable function to approximate $F$ within a preassigned tolerance.

However, by (\ref{fun2}), (\ref{fun3}), (\ref{fun4}), and (\ref{fun5}), we get $F(t)$ is a step function which contains several $0-1$ items. Also, the arguments of $F(t)$ are $x_1^U(t), y_1^U(t),\dots,x_N^U(t), y_N^U(t), x_1^O(t), y_1^O(t), \dots, x_M^O(t)$, $y_M^O(t)$ in fact.
Let
\begin{eqnarray*}
&&\mathbf{X^{U}}(t):= \big(x_1^U(t), y_1^U(t),\dots,x_N^U(t), y_N^U(t)\big)^{\top},\\
&&\mathbf{X^{O}}(t):= \big(x_1^O(t), y_1^O(t),\dots,x_M^O(t), y_M^O(t)\big)^{\top},\\
&&\mathbf{X^{T}}(t):= \big(x_{\rm{tar},1}^T(t), y_{\rm{tar},1}^T(t),\dots,x_{\rm{tar},N}^T(t), y_{\rm{tar},N}^T(t)\big)^{\top},\\
&&\mathbf{X^{U,O}}(t):=\big(\mathbf{X^{U}}(t)^{\top},\mathbf{X^{O}}(t)^{\top}\big)^{\top},\\
&&\mathbf{X}(t):= \big(\mathbf{X^{U,O}}(t)^{\top},\mathbf{X^{T}}(t)^{\top}\big)^{\top}.
\end{eqnarray*}
Then, from (\ref{fun5}) and  (\ref{newobjfun}) the functions $F(t)$ and $H(t)$ can be respectively expressed by $F(t)=F\left(\mathbf{X^{U,O}}(t)\right)$ and
\begin{equation}\label{Hmatrixform}
H(t)=H\left(\mathbf{X}(t)\right)=F\left(\mathbf{X^{U,O}}(t)\right)+\frac{1}{2} \left(\mathbf{X^{U}}(t)-\mathbf{X^{T}}(t)\right)^{\top}\left(\mathbf{X^{U}}(t)-\mathbf{X^{T}}(t)\right).
\end{equation}

To make $F\left(\mathbf{X^{U,O}}(t)\right)$ continuous, we use the polygonal function to approximate $F\left(\mathbf{X^{U,O}}(t)\right)$ at these jump discontinuities.
Take the signum function
$$
\mbox{sign}(x) =\left\{
\begin{aligned}
	-1,  &~~ x<0\\
	0,  &~~ x=0\\
	1, &~~ x>0
\end{aligned}
\right.
$$
as an example, which can be approximated by the following polygonal function:
$$
g(x) =\left\{
\begin{aligned}
	-1,  &~~ x\leq -\eta\\
	x/\eta,  &~~ -\eta<x<\eta\\
	1, &~~ x\geq \eta
\end{aligned}
\right.,
$$
where $\eta>0$ is a small real number.
The approximation effect of $g(x)$ is shown in Figure \ref{Fig:step}, and $g(x)$ is a continuous function.
\begin{figure}[!t]
\centering
\includegraphics[height=2in, width=2.5in]{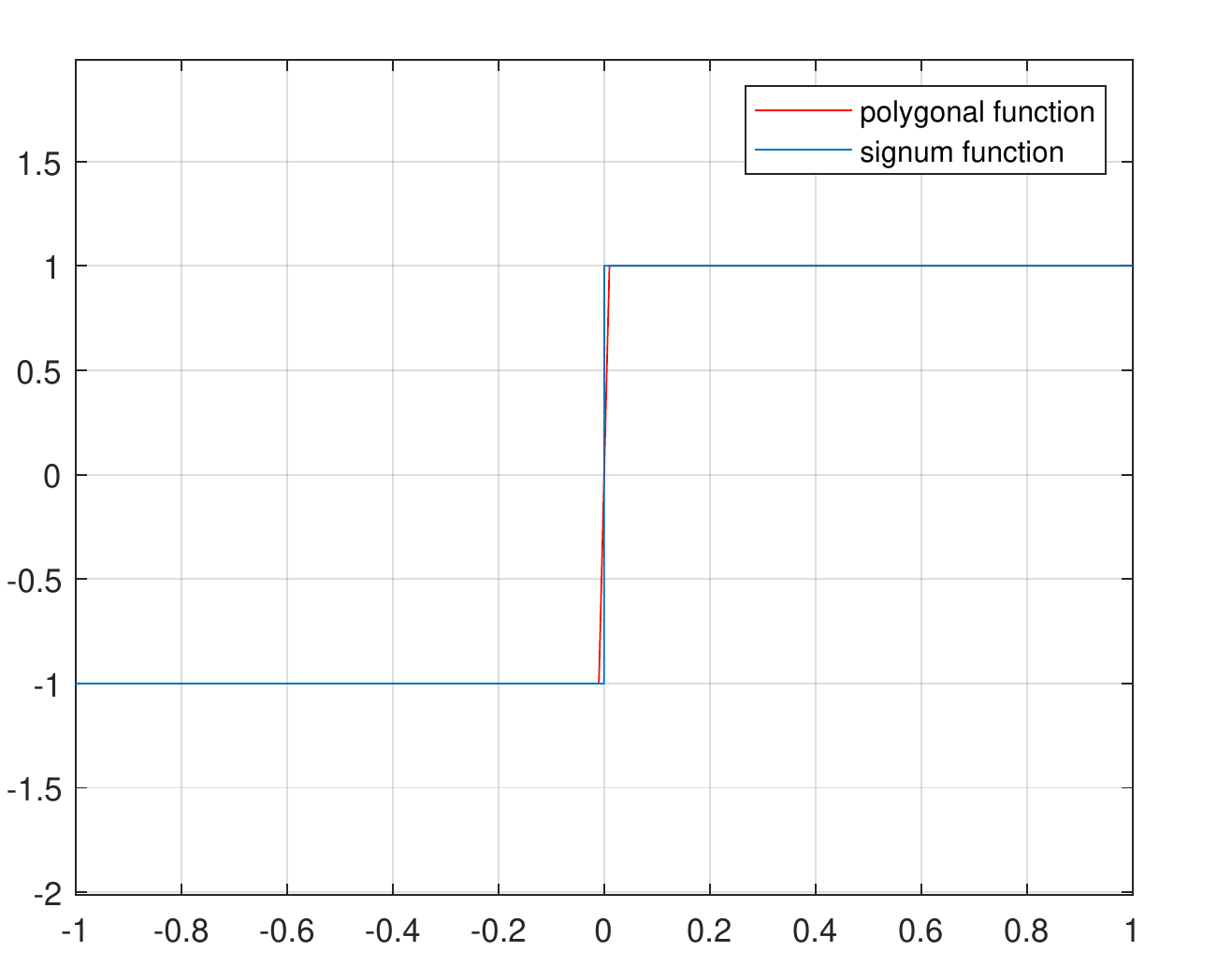}
\caption{The approximation effect of $g(x)$.}
\label{Fig:step}
\end{figure}

With this continuous method, we can obtain a continuous function $\tilde{F}\left(\mathbf{X^{U,O}}(t)\right)$ which is the continuous version of $F\left(\mathbf{X^{U,O}}(t)\right)$.
However, $\tilde{F}\left(\mathbf{X^{U,O}}(t)\right)$ is still not differentiable.
The back propagation (BP) neural network is one of the most popular learning algorithms in neural networks\cite{hecht1992theory}. Due to the capabilities of nonlinear mapping and generalization, the BP neural network is commonly used in the field of function approximation.
We adopt the single hidden layer BP neural network, which includes an input layer, a hidden layer, and an output layer.
The training, validation and test errors are evaluated using mean square error (MSE). If the best validation performance is in the allowance range, we take $\tilde{F}(\mathbf{X^{U,O}}(t))$ as the continuous and differentiable version of $F(\mathbf{X^{U,O}}(t))$.
Next, we introduce how to construct the three layers of the BP neural network for our approximation problem:

\emph{Input layer}:  We use the dynamic positions of agents and radar-missiles as the input variables of this BP neural network. In the training phase, we randomly generate their dynamic positions as training data. Then, the input vector is
\begin{equation*}
\mathbf{X^{U,O}}:= \big(x_1^U, y_1^U,\dots,x_N^U, y_N^U,
x_1^O, y_1^O, \dots, x_M^O, y_M^O\big)^{\top}.
\end{equation*}
The agents may break down or be shot down, consequently, we may lose some input information.
However, the size of input layer is fixed as $2N+2M$.
As a result, we adopt the zero-padding method to meet the requirement of the fixed input size of the  BP neural network, in which zeros are padded to the lost input information\cite{hashemi2019enlarging}, that is, when an agent breaks down or is shot down, we may lose its position information and the corresponding input is padded with zero.

\emph{Hidden layer:}
We choose the sigmoidal function $$\sigma(x) = \frac{1}{1+ \exp(-x)}$$ as the activation function of the hidden layer in the BP neural network. Assume that the neuron number of the hidden layer is $J$.
Let $\mathbf{\omega^j}\in\mathbb{R}^{2N+2M}$ and $\theta_j\in\mathbb{R}$, $j=1,\ldots,J$ be the weights and biases between the input layer and the hidden layer respectively.
Then, the output of the hidden layer is
$$\left(\sigma((\mathbf{X^{U,O}})^{\top}\mathbf{\omega^1}+\theta_1),\ldots,\sigma((\mathbf{X^{U,O}})^{\top}\mathbf{\omega^J}+\theta_J)\right).$$

\emph{Output layer:}
We choose the activation function of the output layer to be a linear function. Let $\lambda_j\in\mathbb{R}, j=1,\ldots,J$, and $\mu\in\mathbb{R}$ be the weights and bias between the hidden layer and the output layer
respectively.
Then, the output information is
\begin{equation} \label{aprfun}
	F^*(\mathbf{X^{U,O}}) :=\mu+ \sum_{j=1}^J \lambda_j \sigma((\mathbf{X^{U,O}})^{\top}\mathbf{\omega^j}+\theta_j) \nonumber=\mu+ \sum_{j=1}^J  \frac{\lambda_j}{1+ \exp(-(\mathbf{X^{U,O}})^{\top}\mathbf{\omega^j}-\theta_j)} .
\end{equation}

$F^*(\mathbf{X^{U,O}})$ can approximate $\tilde{F}(\mathbf{X^{U,O}})$ arbitrarily closely with suitable parameters $\mu, J, \lambda_j, \mathbf{\omega^j}, \theta_j$.
We can randomly generate a large number of input data set $\mathcal{X}=\{\mathbf{X^{U,O}}\}$, and choose
$\Big\{\left(\mathbf{X^{U,O}}, \tilde{F}(\mathbf{X^{U,O}})\right): \mathbf{X^{U,O}}\in \mathcal{X}\Big\}$ to be the training set.
After training, we can get our desired parameters $\mu, J, \lambda_j, \mathbf{\omega^j}, \theta_j$ of the neural network. Therefore, we can find a differentiable function $F^*$ to approximate $\tilde{F}$ and the gradient of $F^*$ can be calculated.
Replacing $F(t)$ in (\ref{opt_pro}) with $F^*(t)$,   the online optimal path control problem for the agents in USs can be approximated as follows:
\begin{flalign}\label{appro_fuc}
    &\min_{\big\{\left(x_i^U(t), y_i^U(t)\right)\big\}_{i=1}^N} \left\{ F^*(t)+\sum_{i=1}^{N}f_{i,1}(t)\right\}\nonumber\\
	&~\mbox{s.t.}~~\left\{
	\begin{aligned}
		&  v_i^U(t)  \leq v_{\rm{max}}\\
		& |\dot{\psi}_i(t)| \leq \frac{n_{\rm{max}}\emph{g}}{v_i^U(t)}
	\end{aligned}
	\right., ~~i=1,\ldots,N.
\end{flalign}

\begin{remark}
The parameter training of the neural network does not depend on the real-time positions of the agents and radar-missiles, and the computation can be off-line. With the trained parameters of the neural network, we can solve
problem (\ref{appro_fuc}) online.
\end{remark}

\subsection{Algorithm for solving optimal path control problem based on RHC}\label{Subsec_plan}
With the principle of RHC, we try to solve the optimal control problem (\ref{appro_fuc}) in real-time with modified gradient optimization techniques.
Gradient descent\cite{barzilai1988two} is one of the simplest and most classic methods for solving optimization problems. This method requires less computational effort and shorter execution time, which is beneficial to design a real-time algorithm. We note that the arguments of our objective function are the positions of  agents, targets and radar-missiles and we can only control the the agents' movement. Therefore, our algorithm only do gradient descent operations for part arguments with the first two constraints of (\ref{constraint}). In particular, we first compute the negative gradient of the objective function of problem (\ref{appro_fuc}) with respect to $\mathbf{X^{U}}$.
Let
\begin{equation}\label{obfun1}
H^*(t):=F^*(t)+\sum_{i=1}^{N}f_{i,1}(t).
\end{equation}
The negative gradient of $H^*(t)$ with respect to each agent $U_i$ is
\begin{equation}\label{direct1}
\mathbf{d}^i(t)=\left(d_x^i(t),d_y^i(t)\right):=-\left( \frac{\partial H^*}{\partial x_i^U}(t), \frac{\partial H^*}{\partial y_i^U}(t)\right)= -\left(\frac{\partial F^*}{\partial x_i^U}(t)+\frac{\partial f_{i,1}}{\partial x_i^U}(t), \frac{\partial F^*}{\partial y_i^U}(t)+\frac{\partial f_{i,1}}{\partial y_i^U}(t)\right),
\end{equation}
where the last equal sign uses (\ref{obfun1}) and (\ref{fun1}). Let
$\mathbf{\omega^j}=(\omega_1^j,\omega_2^j,\ldots,\omega_{2N+2M}^j)$.
By (\ref{aprfun}), (\ref{fun1}), and (\ref{direct1}), we have
\begin{equation*}
	d_x^i(t)= -\left(\sum_{j=1}^J  \frac{\lambda_j \omega_{2i-1}^j \exp(-\mathbf{X^{U,O}}(t)^{\top}\mathbf{\omega^j}-\theta_j)}{[1+ \exp(-\mathbf{X^{U,O}}(t)^{\top}\mathbf{\omega^j}-\theta_j)]^2}\right)-\left(x_i^U(t) - x_{\rm{tar},i}^T(t)\right),~~i=1,\ldots,N,
\end{equation*}
and
\begin{equation*}
	d_y^i(t)= -\left(\sum_{j=1}^J  \frac{\lambda_j \omega_{2i}^j \exp(-\mathbf{X^{U,O}}(t)^{\top}\mathbf{\omega^j}-\theta_j)}{[1+ \exp(-\mathbf{X^{U,O}}(t)^{\top}\mathbf{\omega^j}-\theta_j)]^2}\right) -\left(y_i^U(t) - y_{\rm{tar},i}^T(t)\right),~~i=1,\ldots,N.
\end{equation*}

However, the heading angle variation rate has the limit
$$|\dot{\psi}_i(t)| \leq \frac{n_{\rm{max}}\emph{g}}{v_{\rm{max}}}.$$
We denote the maximum heading angle variation during each execution time $\Delta t $ by
$$\psi_{\max}:= \frac{n_{\rm{max}}\emph{g}}{v_{\rm{max}}} \cdot \Delta t .$$
Let $\tilde{\theta}_i(t)$ be the angle variation between the current velocity $\left(\dot{x}_i^U(t),\dot{y}_i^U(t)\right)$ and the negative gradient $\mathbf{d}^i(t)$, i.e.,
\begin{equation*}
\tilde{\theta}_i(t)=\arccos\left( \frac{\dot{x}_i^U(t)d_x^i(t)+\dot{y}_i^U(t)d_y^i(t) }{\sqrt{[\dot{x}_i^U(t)^2+\dot{y}_i^U(t)^2][d_x^i(t)^2+d_y^i(t)^2]}}\right).
\end{equation*}

When the heading angle variation $\tilde{\theta}_i(t)$ is smaller than or equal to $\psi_{\max}$, we choose $\mathbf{d}^i(t)$ as the new direction of agent $U_i$. Otherwise, we let agent $U_i$ turn to the negative gradient as much as possible.
As stated in\cite{ritter1991fast}, the sign of the cross-product of two vectors in the coordinate plane can help us to determine whether the angle between two vectors is clockwise or counterclockwise.
If
$$\left(\dot{x}_i^U(t),\dot{y}_i^U(t)\right) \times \mathbf{d}^i(t) =\dot{x}_i^U(t) d_y^i(t) - d_x^i(t)\dot{y}_i^U(t)>0,$$
the negative gradient $\mathbf{d}^i(t)$ is in the counterclockwise direction of the current velocity, we turn agent $U_i$ in the counterclockwise direction by $\psi_{\max}$, otherwise we turn agent $U_i$ in the clockwise direction.

Let $\mathbf{\hat{d}}^i(t) = \left(\hat{d}^i_x(t),\hat{d}^i_y(t)\right)$ be the new direction. We have
\begin{equation*}
\hat{d}^i_x(t):= \left\{
\begin{aligned}
	&\dot{x}_i^U(t) \mbox{cos}\psi_{\max}-\dot{y}_i^U(t)\mbox{sin}\psi_{\max} , ~~\mbox{if}~\left(\dot{x}_i^U(t),\dot{y}_i^U(t)\right) \times \mathbf{d}^i(t)>0\\
	&\dot{x}_i^U(t) \mbox{cos}\psi_{\max}+\dot{y}_i^U(t)\mbox{sin}\psi_{\max} , ~~\text{otherwise}
\end{aligned}
\right.,
\end{equation*}
and
\begin{equation*}
\hat{d}^i_y(t):= \left\{
\begin{aligned}
	&\dot{x}_i^U(t)\mbox{sin}\psi_{\max} +\dot{y}_i^U(t)\mbox{cos}\psi_{\max} , ~~\mbox{if}~\left(\dot{x}_i^U(t),\dot{y}_i^U(t)\right) \times \mathbf{d}^i(t)>0\\
	&-\dot{x}_i^U(t)\mbox{sin}\psi_{\max}+\dot{y}_i^U(t)\mbox{cos}\psi_{\max} ,
	~~\text{otherwise}
\end{aligned}
\right..
\end{equation*}

Let $\mathbf{\widetilde{d}}^i(t) := (\widetilde{d}^i_x(t),\widetilde{d}^i_y(t)  )$ be the unit vector of the final direction, i.e.,
\begin{equation}\label{unitdirect}
	\mathbf{\widetilde{d}}^i(t):=
	\left\{
	\begin{aligned}
		& \mathbf{d}^i(t)/\|\mathbf{d}^i(t)\|, ~~\mbox{if}~~\tilde{\theta}_i(t)\leq \psi_{\max}\\
		&\mathbf{\hat{d}}^i(t)/\|\mathbf{\hat{d}}^i(t)\| ,~~\text{otherwise}
	\end{aligned}
	\right.,
\end{equation}
where $\|\cdot\|$ denotes the $L_2$-norm.
 Using the method of gradient descent\cite{barzilai1988two},
the positions of all agents are updated by
\begin{flalign}\label{update}
 \left\{
\begin{aligned}
x_i^U(t+\Delta t) = x_i^U(t) +  \kappa_i(t)  \widetilde{d}^i_x(t)\\
y_i^U(t+\Delta t) = y_i^U(t) +  \kappa_i(t)  \widetilde{d}^i_y(t)
\end{aligned}
\right., ~i=1,\ldots,N, t\geq 0,
\end{flalign}
where $\kappa_i(t)>0$ is the step size.
Since the exact solution of the best step size is complicated and needs to consider the constraint on the maximum velocity, we do not want to make too much effort to find it.
Therefore, we suppose that all agents move at their maximum velocity and the step size $\kappa_i(t)$ is chosen as the maximum distance of motion between every two steps, i.e.,
\begin{equation*}
\kappa_i(t)=v_{\rm{max}} \Delta t.
\end{equation*}
The agents keep running (\ref{update}) to update their positions, that is, $D_i^{U-T}(t) \leq v_{\rm{max}} \Delta t,~i=1,\dots,N$. When an agent reach a target, that is, $D_i^{U-T}(t) \leq v_{\rm{max}}$, we consider that this agent has accomplished its task. Then, we remove it and its target (and copied targets) from our algorithm. Other agents keep running (\ref{update}) until all the targets are reached.
As a summary, we draw the following flow chart Figure \ref{Fig:flowchart} to show the process of optimal control for the agents in USs.
\begin{figure}[!t]
\centering
\includegraphics[height=2.2in, width=3.5in]{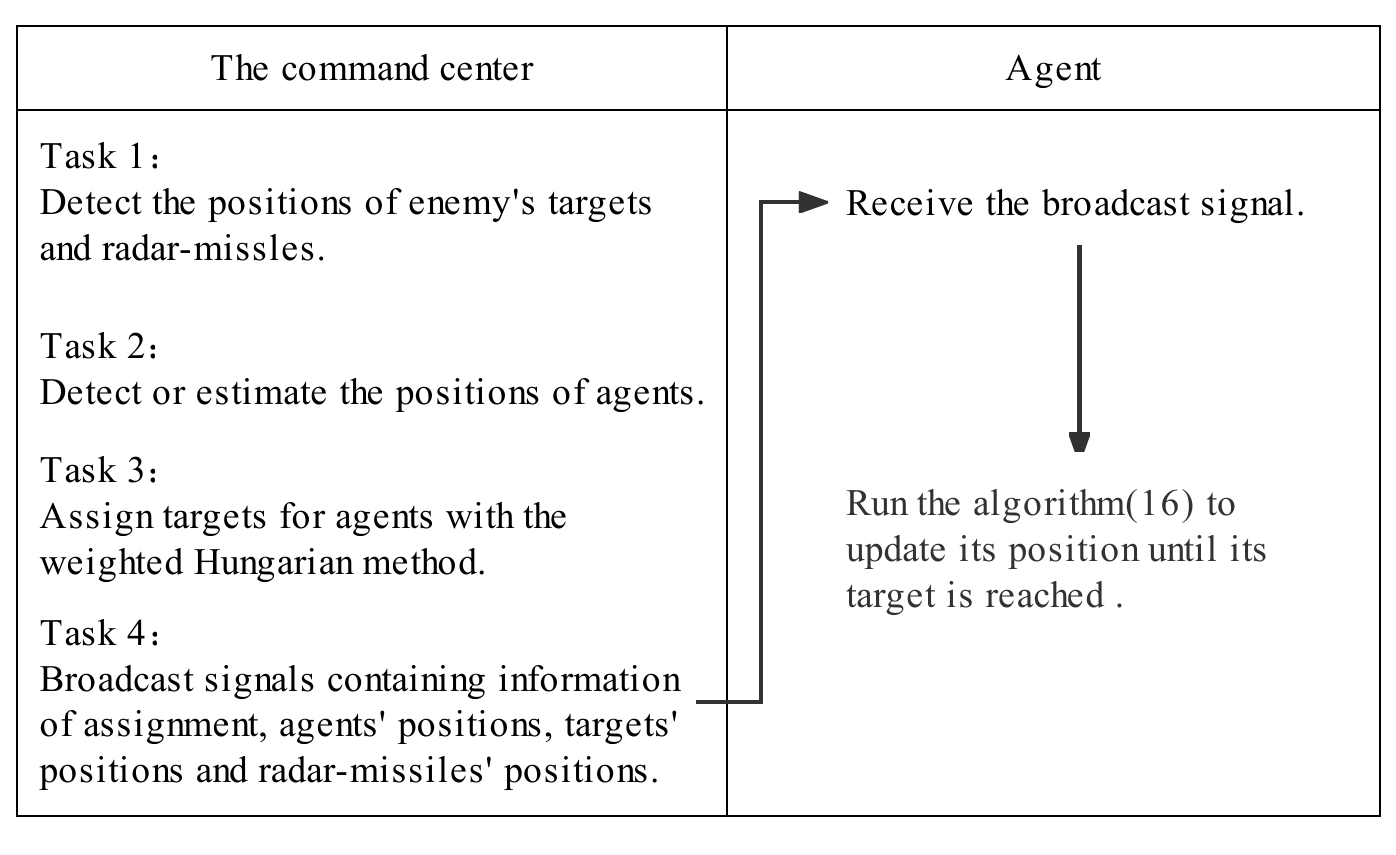}
\caption{Flow chart of the optimal control for the agents.}
\label{Fig:flowchart}
\end{figure}

\begin{remark}
The agents in unmanned systems may perform tasks in a $3$-dimensional environment. To simplify the expression, we only consider a $2$-dimensional flat plane. However, it is not difficult to extend our method to $3$-dimensional environment by enlarging the dimensions of the variables and  adjusting the model.
\end{remark}

\subsection{Optimal control for the agents with unknown positions}\label{unknown}
In this section, we first solve Case II mentioned in Section \ref{description}. The command center cannot know the exact positions of agents it commands, but it can calculate and estimate their approximate positions according to their initial positions and rough movements of each step.  Since the estimated positions may not be accurate, each agent may send its position to the command center on a long timed interval to enhance accuracy.
On the other hand, our algorithm does not require high accuracy of agents' positions, except for the collision avoidance part. Thus, each agent needs to be equipped with a detector to avoid the collision.
If an agent detects any obstacles or other agents in its surroundings with its sensor or radar, it will avoid collisions by steering into opposite direction or other techniques.
We make a small change to the objective function. Let
\begin{flalign*}
	\gamma_{i,j}(t):= \left\{
	\begin{aligned}
	1\quad,  &\quad \text{if $U_i$ detects $U_j$ at time $t$}\\
	0\quad,  & \quad \text{otherwise}
	\end{aligned}
	\right..
\end{flalign*}
Similar to section \ref{Subsec_plan}, the command center can calculate $U_i$'s direction $\left(\widetilde{d}^i_x(0),\widetilde{d}^i_y(0)\right)$. Therefore, the approximate position of $U_i$ at next step is
\begin{flalign*}
 \left\{
\begin{aligned}
x_i^{*U}(\Delta t) = x_{i}^U(0) +  \kappa_i(0)  \widetilde{d}^i_x(0)\\
y_i^{*U}(\Delta t) = y_{i}^U(0) +  \kappa_i(0)  \widetilde{d}^i_y(0)
\end{aligned}
\right.,
\end{flalign*}
where $\left(x_{i}^U(0),y_{i}^U(0)\right)$ is the initial position of $U_i$.
We use the the approximate position $\left(x_i^{*U}(t),y_i^{*U}(t)\right)$ to replace the accurate position $\left(x_i^U(t),y_i^U(t)\right)$ in section \ref{algorithm}.
In this way, the approximate positions of all agents can be updated by
\begin{flalign}\label{update2}
 \left\{
\begin{aligned}
x_i^{*U}(t+\Delta t) = x_i^{*U}(t) +  \kappa_i(t)  \widetilde{d}^i_x(t)\\
y_i^{*U}(t+\Delta t) = y_i^{*U}(t) +  \kappa_i(t)  \widetilde{d}^i_y(t)
\end{aligned}
\right., ~i=1,\ldots,N, t\geq 0.
\end{flalign}
The command center can get the approximate positions of all agents with (\ref{update2}) and use the weighted Hungarian method to assign targets for them. Then, the command center broadcast the information of assignment and the positions of targets and radar-missiles. The agents receive broadcast signals and run (\ref{update2}) to update their positions.

\section{The convergence Analysis}\label{analysis}
In this section, we give our theoretic analysis of our modified gradient optimization algorithm (\ref{update}).
To simplify the exposition, we choose $\Delta t=1$ without loss of generality.
Also, as we assumed in Section \ref{description}, the targets can move bounded by a maximum velocity. Assume that $\delta>0$ is the maximum velocity of targets.
Let
$$\delta(k):=\mathbf{X}^{\mathbf{T}} (k+1) - \mathbf{X}^{\mathbf{T}} (k)$$
represent the movement of the targets at time step $k$ . We note that $\delta(k)$ is a $2N$-dimensional vector and its $(2i-1)$-th and $(2i)$-th elements represent the movement of $U_i$'s target, then
\begin{equation}\label{newdecre1}
\delta_{2i-1}^2(k)+ \delta_{2i}^2(k) \leq \delta^2 ,~i=1,\dots,N.
\end{equation}

To simplify the analysis, we assume that  the angle variations  $\tilde{\theta}_1(k),\ldots,\tilde{\theta}_N(k)$
are always not bigger than the maximum heading angle variation $\psi_{\max}$,
which means by (\ref{unitdirect}), the unit vector of the iteration direction
\begin{equation*}\label{newdirect}
\mathbf{\widetilde{d}}^i(k)=\mathbf{d}^i(k)/\|\mathbf{d}^i(k)\|.
\end{equation*}
The convergence result for our algorithm can be formulated as the following theorem.
\begin{thm}\label{thm_2}
Suppose that $v_{\rm{max}}>\sqrt{2}\delta$, and there exists a constant $b \in (0,[v_{\rm{max}}^2-2\delta^2]/[2\sqrt{2}v_{\rm{max}}])$ such that\\
\begin{equation}\label{thm_condition}
\max_{i,k}\left\{\Big |\frac{\partial F^*}{\partial x_i^U}(k)\Big |,\Big |\frac{\partial F^*}{\partial y_i^U}(k)\Big | \right\}\leq b,
\end{equation}
and
\begin{equation}\label{thm_2_1}
\left( 1- 2\sqrt{1-\frac{2\sqrt{2}b}{v_{\rm{max}}}}-\frac{2\sqrt{2}\delta}{v_{\rm{max}}}\right)v_{\rm{max}}^2\\+ \delta^2 +2\sqrt{2} v_{\rm{max}} \delta +4 v_{\rm{max}} b:=-\varepsilon  <0.
\end{equation}
Then, under the iteration (\ref{update}), each agent $U_i$ can reach its target within time
$\left(D_{i}^{U-T}(0)\right)^2/\varepsilon$.
\end{thm}
\begin{proof}
Let $k_i^*$ be the stop time that agent $U_i$ reaches its assigned target, i.e.,
when $k<k_i^*$, agent $U_i$ cannot reach its target which indicates $D_{i}^{U-T}(k) > v_{\rm{max}}$ ; at time $k_i^*$, $D_{i}^{U-T}(k) \leq v_{\rm{max}}$ and agent $U_i$ reaches its target.
We consider the case when $k<k_i^*$.
Let
\begin{flalign*}
 \left\{
\begin{aligned}
&x_i^{U-T}(k):=x_{i}^{U}(k)-x_{\rm{tar},i}^{T}(k)\\
&y_i^{U-T}(k):=y_{i}^{U}(k)-x_{\rm{tar},i}^{T}(k)
\end{aligned}
\right..
\end{flalign*}
By (\ref{update}),
the decrease of $\left(D_{i}^{U-T}\right)^2$ is
\begin{eqnarray}\label{decrease}
&&\left(D_{i}^{U-T}(k+1)\right)^2-\left(D_{i}^{U-T}(k)\right)^2\\
&&= v_{\rm{max}}^2 + \delta_{2i-1}^2(k)+ \delta_{2i}^2(k) -2\frac{v_{\rm{max}}}{\|\mathbf{d}^{i}(k)\|}\left[d_x^i(k)\delta_{2i-1}(k)+ d_y^i(k)\delta_{2i}(k) \right]\nonumber\\
&&~~~+2\left( \frac{v_{\rm{max}}}{\|\mathbf{d}^{i}(k)\|}d_x^i(k)- \delta_{2i-1}(k) \right)x_i^{U-T}(k)+2\left( \frac{v_{\rm{max}}}{\|\mathbf{d}^{i}(k)\|}d_y^i(k)- \delta_{2i}(k) \right)y_i^{U-T}(k).\nonumber
\end{eqnarray}
By (\ref{direct1}), we have
\begin{flalign}\label{direct_p}
 \left\{
\begin{aligned}
&x_i^{U-T}(k)=-\frac{\partial F^*}{\partial x_i^U}(k)-d_x^i(k)\\
&y_i^{U-T}(k)=-\frac{\partial F^*}{\partial y_i^U}(k)-d_y^i(k)
\end{aligned}
\right..
\end{flalign}
Substituting (\ref{direct_p}) into (\ref{decrease}), we obtain
\begin{eqnarray}\label{newdecre}
&&\left(D_{i}^{U-T}(k+1)\right)^2-\left(D_{i}^{U-T}(k)\right)^2\\
&&=v_{\rm{max}}^2 -2 v_{\rm{max}} \|\mathbf{d}^{i}(k)\| + \delta_{2i-1}^2(k)+ \delta_{2i}^2(k) \nonumber\\
&&~~~- 2\frac{v_{\rm{max}}}{\|\mathbf{d}^{i}(k)\|}\left[d_x^i(k)\delta_{2i-1}(k)+ d_y^i(k)\delta_{2i}(k) \right]\nonumber\\
&&~~~-2\Big[\delta_{2i-1}(k) x_i^{U-T}(k)+ \delta_{2i}(k)y_i^{U-T}(k) \Big]\nonumber\\
&&~~~-2\frac{v_{\rm{max}}}{\|\mathbf{d}^{i}(k)\|}\left[d_x^i(k)\frac{\partial F^*}{\partial x_{i}^{U}}(k)+d_y^i(k)\frac{\partial F^*}{\partial y_{i}^{U}}(k) \right].\nonumber
\end{eqnarray}
We now compute each part of the right of (\ref{newdecre}).
For the third line  in (\ref{newdecre}), we obtain
\begin{equation}\label{newdecre2}
- 2\frac{v_{\rm{max}}}{\|\mathbf{d}^{i}(k)\|}\left[d_x^i(k)\delta_{2i-1}(k)+ d_y^i(k)\delta_{2i}(k) \right] \leq  2v_{\rm{max}}\cdot \delta \cdot \frac{|d^i_x(k)|+|d^i_y(k)|}{\|\mathbf{d}^{i}(k)\|} \leq  2\sqrt{2} v_{\rm{max}} \delta.
\end{equation}
For the  fourth line in (\ref{newdecre}), we have
\begin{equation}\label{newdecre3}
-2\Big[\delta_{2i-1}(k) x_i^{U-T}(k)+ \delta_{2i}(k)y_i^{U-T}(k) \Big]\leq 2 \delta \Big[|x_i^{U-T}(k)|+|y_i^{U-T}(k)| \Big]\leq 2\sqrt{2}\delta D_{i}^{U-T}(k) .
\end{equation}
For the last line in (\ref{newdecre}), by (\ref{thm_condition}) we get
\begin{equation}\label{newdecre4}
-2\frac{v_{\rm{max}}}{\|\mathbf{d}^{i}(k)\|}\left[d_x^i(k)\frac{\partial F^*}{\partial x_{i}^{U}}(k)+d_y^i(k)\frac{\partial F^*}{\partial y_{i}^{U}}(k) \right]\leq 2 v_{\rm{max}} \left[\Big|\frac{\partial F^*}{\partial x_{i}^{U}}(k)\Big|+\Big|\frac{\partial F^*}{\partial y_{i}^{U}}(k)\Big| \right] \leq 4 v_{\rm{max}} b.
\end{equation}
Substituting (\ref{newdecre1}) and (\ref{newdecre2})-(\ref{newdecre4}) into (\ref{newdecre}), we have
\begin{equation}\label{newdecre_new}
\left(D_{i}^{U-T}(k+1)\right)^2-\left(D_{i}^{U-T}(k)\right)^2\leq v_{\rm{max}}^2 -2 v_{\rm{max}} \|\mathbf{d}^{i}(k)\| + \delta^2 +2\sqrt{2} v_{\rm{max}} \delta+2\sqrt{2}\delta D_{i}^{U-T}(k)+4 v_{\rm{max}} b.
\end{equation}
Also, by (\ref{direct1}), we get
\begin{eqnarray*}
\|\mathbf{d}^{i}(k)\|^2&&=\left(d_x^i(k)\right)^2+\left(d_y^i(k)\right)^2 \\
&&=\left(D_{i}^{U-T}(k)\right)^2+ \left(\frac{\partial F^*}{\partial x_{i}^{U}}(k)\right)^2+\left(\frac{\partial F^*}{\partial y_{i}^{U}}(k) \right)^2+2\left(\frac{\partial F^*}{\partial x_{i}^{U}}(k)\right)x_i^{U-T}(k)+2\left(\frac{\partial F^*}{\partial y_{i}^{U}}(k)\right)y_i^{U-T}(k).
\end{eqnarray*}
Together this with (\ref{thm_condition}), $D_{i}^{U-T}(k) > v_{\rm{max}}$ and $b \in (0,[v_{\rm{max}}^2-2\delta^2]/[2\sqrt{2}v_{\rm{max}}])$ in Theorem \ref{thm_2}, we obtain
\begin{eqnarray*}
\|\mathbf{d}^{i}(k)\|^2 &&\geq \left(D_{i}^{U-T}(k)\right)^2  -2b\Big[|x_i^{U-T}(k)|+|y_i^{U-T}(k)|\Big]\\
&&\geq \left(D_{i}^{U-T}(k)\right)^2 -2\sqrt{2}b D_{i}^{U-T}(k),
\end{eqnarray*}
which indicates
$$\|\mathbf{d}^{i}(k)\|\geq \sqrt{1-\frac{2\sqrt{2}b}{D_{i}^{U-T}(k)}}D_{i}^{U-T}(k).$$
Therefore, (\ref{newdecre_new}) can be bounded by
\begin{eqnarray}\label{newdecre_nn}
&&\left(D_{i}^{U-T}(k+1)\right)^2-\left(D_{i}^{U-T}(k)\right)^2\\
&&\leq v_{\rm{max}}^2 -\left(2 v_{\rm{max}} \sqrt{1-\frac{2\sqrt{2}b}{D_{i}^{U-T}(k)}}-2\sqrt{2}\delta\right)D_{i}^{U-T}(k) + \delta^2 +2\sqrt{2} v_{\rm{max}} \delta +4 v_{\rm{max}} b.\nonumber\\
&&\leq v_{\rm{max}}^2 -\left(2 v_{\rm{max}} \sqrt{1-\frac{2\sqrt{2}b}{v_{\rm{max}}}}-2\sqrt{2}\delta\right)v_{\rm{max}} + \delta^2 +2\sqrt{2} v_{\rm{max}} \delta +4 v_{\rm{max}} b \nonumber\\
&&=\left( 1- 2\sqrt{1-\frac{2\sqrt{2}b}{v_{\rm{max}}}}-\frac{2\sqrt{2}\delta}{v_{\rm{max}}}\right)v_{\rm{max}}^2 + \delta^2 +2\sqrt{2} v_{\rm{max}} \delta +4 v_{\rm{max}} b = -\varepsilon,\nonumber
\end{eqnarray}
where the second inequality uses the fact $D_{i}^{U-T}(k) > v_{\rm{max}}$ and $b<\frac{v_{\rm{max}}^2-2\delta^2}{2\sqrt{2}v_{\rm{max}}}$,
and the last equation uses (\ref{thm_2_1}).
 Using (\ref{newdecre_nn}) repeatedly, we can get $\left(D_{i}^{U-T}(k_i^*)\right)^2\leq \left(D_{i}^{U-T}(0)\right)^2-k_i^* \varepsilon,$ while $\left(D_{i}^{U-T}(k_i^*)\right)^2>0$. Therefore, we have  $$k_i^* <\frac{\left(D_{i}^{U-T}(0)\right)^2}{\varepsilon},~i=1,\dots,N,$$
which means that each agent $U_i$ can reach its target within time $\left(D_{i}^{U-T}(0)\right)^2/\varepsilon$.
\end{proof}

The condition (\ref{thm_condition}) depends on the system state $\mathbf{X}(k)$. In fact, this condition can be replaced by a parameter condition of neural network which does not depend on the system state.

\begin{corollary}\label{corollary}
Theorem \ref{thm_2} still holds by using the neural network condition
\begin{equation}\label{coro_condi}
\max_{i=1,\dots,N+M}\left\{\sum_{j=1}^J |\lambda_j \omega_{2i-1}^j|,\sum_{j=1}^J |\lambda_j \omega_{2i-1}^j|\right\}\leq 4b
\end{equation}
instead of (\ref{thm_condition}).
\end{corollary}
\begin{proof}
We note that for any $\mathbf{X} \in \mathbb{R}^{2N+2M}$,
\begin{eqnarray*}
\begin{aligned}
\Big|\frac{\partial F^*}{\partial x_i^U}(\mathbf{X})\Big|&=\Big|\sum_{j=1}^J  \frac{\lambda_j \omega_{2i-1}^j \exp(-\mathbf{X}^{\top}\mathbf{\omega^j}-\theta_j)}{[1+ \exp(-\mathbf{X}^{\top}\mathbf{\omega^j}-\theta_j)]^2}\Big| \leq \frac{1}{4} \sum_{j=1}^J |\lambda_j \omega_{2i-1}^j| \leq b, ~~i=1,\dots,N,
\end{aligned}
\end{eqnarray*}
where the last inequality uses (\ref{coro_condi}). Similarly we have $\Big|\frac{\partial F^*}{\partial y_i^U}(\mathbf{X})\Big| \leq b$, so the condition (\ref{thm_condition}) in Theorem \ref{thm_2} still holds.
\end{proof}

\section{Simulation Results}\label{Simulations}
This section presents some simulations to test the algorithm proposed in this paper.
Our offline simulations are carried out on a computer platform with Intel i7 CPU (2.9 GHZ).
We consider an area of $200$km$\times$ $200$km with four radar-missiles. Ten agents need to travel from their initial positions and track five targets. The agents need to avoid collision and antagonistic radar-missiles. The radar detection radius $R_d$ is $10$km and the missile attack radius $R_a$ is $5$km. The secure distance between agents is $100$m.
 It is unnecessary for the agents to change targets frequently, then we call the weighted Hungarian method to assign targets for each agent at the initiation and every five iterations. The  dynamic positions of agents, radar-missiles and targets are updated at each iteration and each iteration is chosen to be $5$s. Each agent is allowed to travel at the maximum velocity of $v_{\rm{max}} = 60$ m/s, therefore we let the step size $\kappa = 0.3$km.  The maximum rate of change of heading angle $n_{\rm{max}}$ is $10$. The  maximum moving range for each agent is $500$km. The penalty factors are separately set as $k_d = k_a = 10^5, k_c = 10^4, k_l = 10^4$, and these parameters satisfy the conditions proposed in Theorem \ref{thm_2}. The balance between these parameters are adjusted and tested empirically.

For the set-up of the BP neural network, a network with 14 input neurons, one output neuron, and one hidden layer with 75 neurons was chosen. The maximum number of training is set as 1000. We randomly generated 100000 sets of dynamic positions $\mathbf{X^{U,O}}$ and calculate their corresponding output value. For the $f_4$ part, we obtain the output according to its historical positions.
Then, we choose $70\%$ of them as training data, $15\%$ of them as validation data, and $15\%$ of them as test data. The training MSE (Mean Square Error) is 0.0002449, which is shown in the Figure \ref{Fig:performance}.

\begin{figure}[!t]
	\centering
	\includegraphics[height=2.2in, width=2.8in]{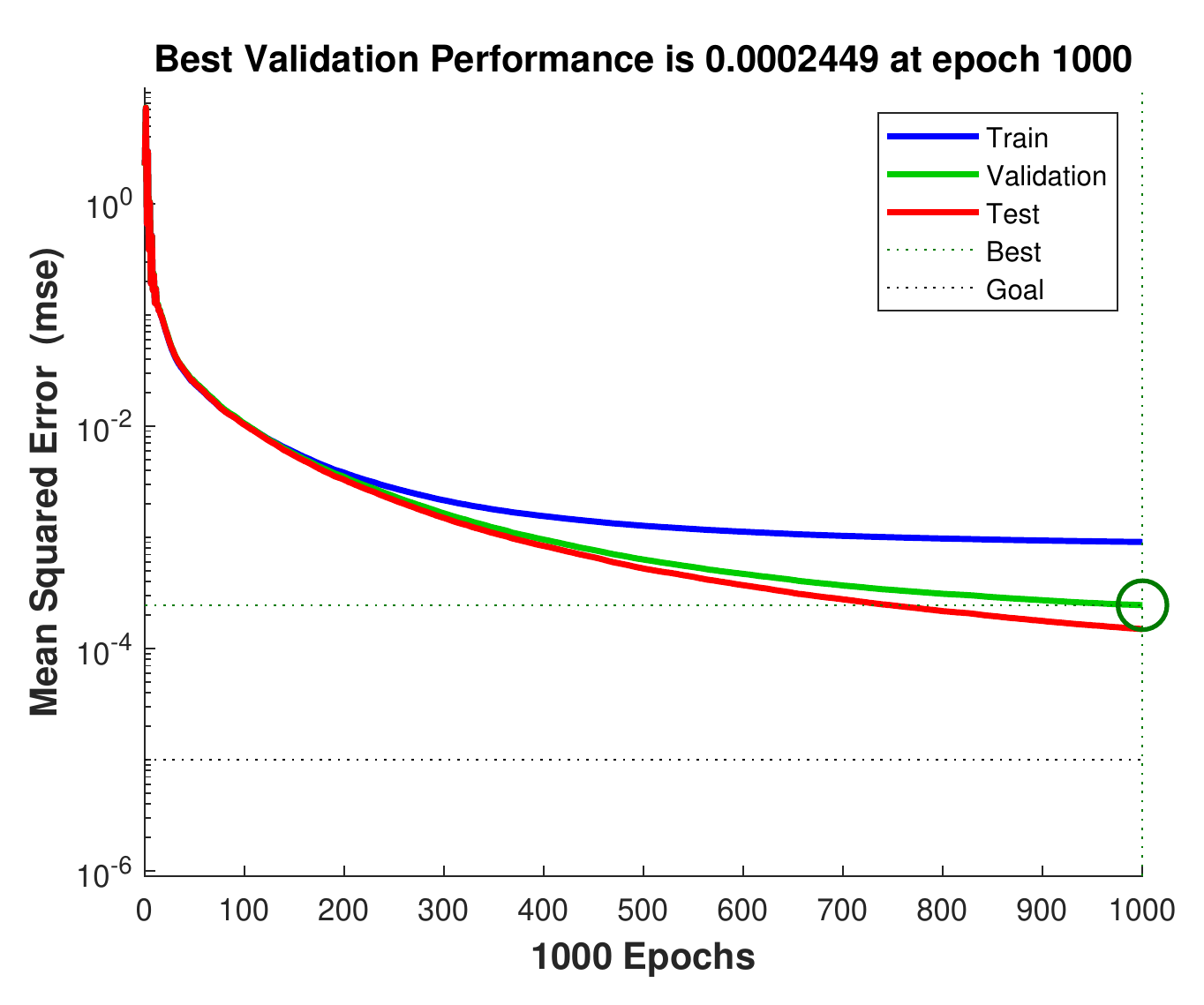}\\
	\caption{The performance of this three-layer BP neural network. }\label{Fig:performance}
\end{figure}

With the well-trained weights of this neural network, we apply our algorithm and conduct several simulations.
First, we compare our algorithm (\ref{update}) with the negative gradient method of the original objective function (\ref{opt_pro}), that is, we use $H$ and $F$ instead of $H^*$ and $F^*$ in Subsection~\ref{Subsec_plan}.
(The derivatives of $F(t)$ in discontinuous points are Dirac Delta functions.) To make the comparison more intuitively, we assume the targets and radar-missiles are static.
In an example of contrast simulations, the paths of agents are shown in Figure \ref{Fig:compare}, where the blue circle represents radar detection threat, and the red circle represents missile attack threat. We use `$\ast$' to represent the initial position of an agent and use `$+$' to represent a target.

\begin{figure}
  \centering
    \subfloat[]{
      \includegraphics[height=2.2in, width=2.2in]{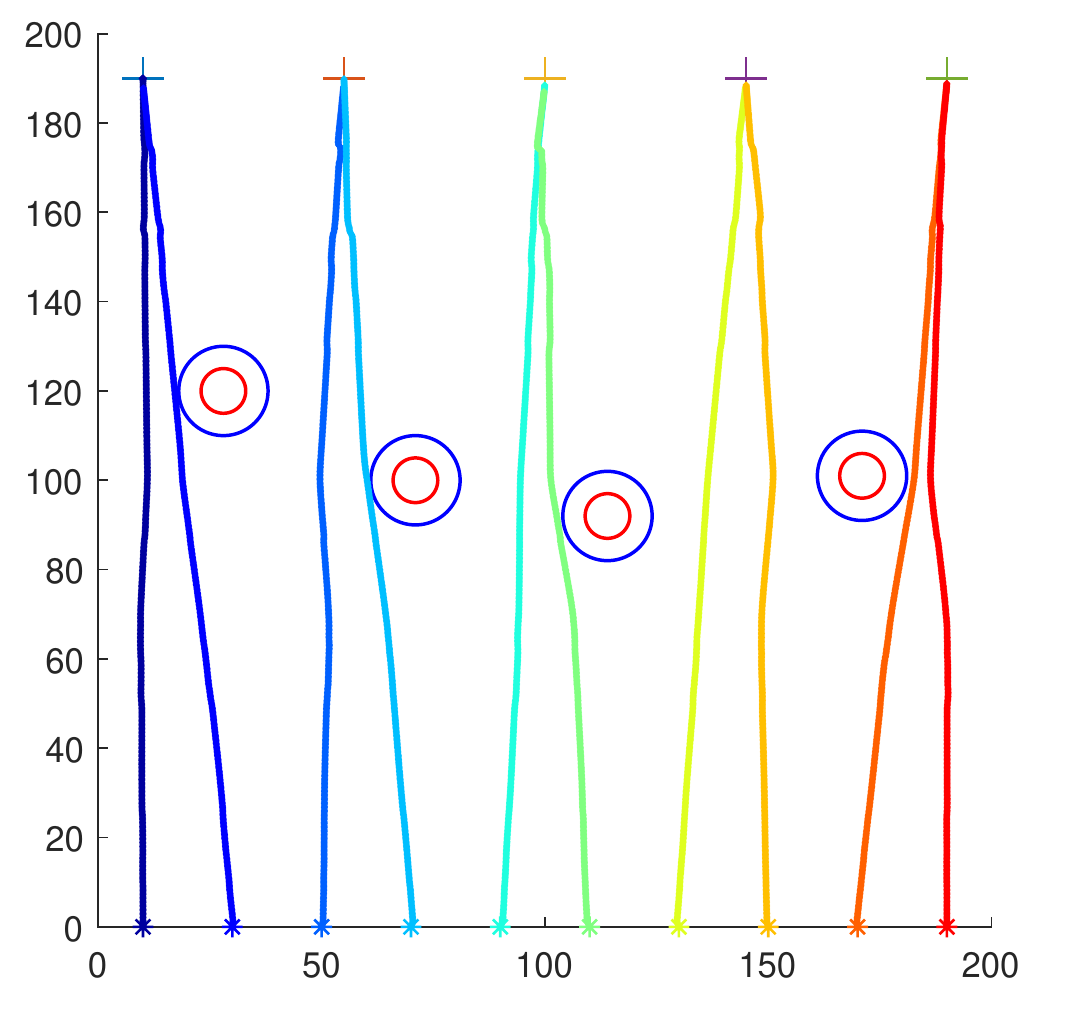}}
    \subfloat[]{
      \includegraphics[height=2.2in, width=2.2in]{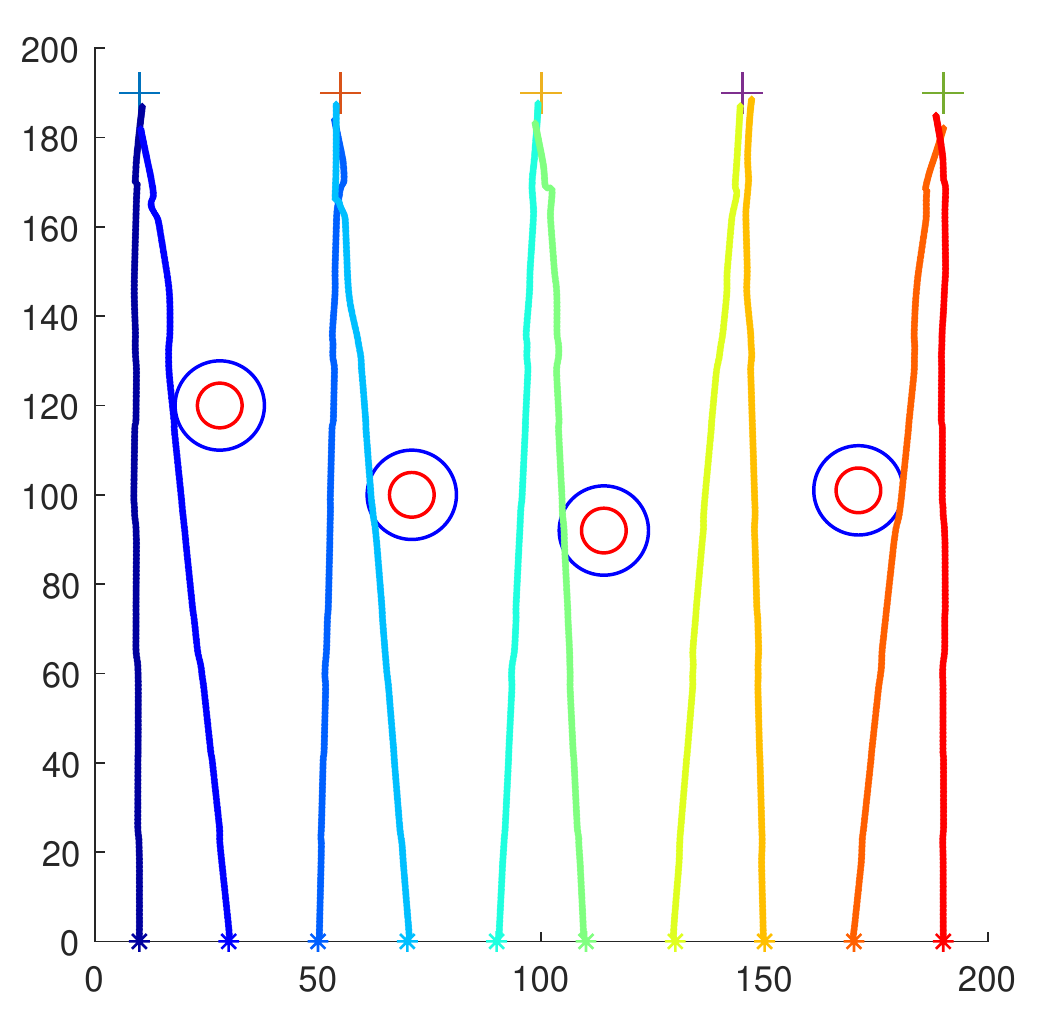}}
\caption{(a) The paths of agents with our algorithm (\ref{update}); (b) The paths of agents with the negative gradient method of the original objective function (\ref{newobjfun}).}
\label{Fig:compare}
\end{figure}

Figure \ref{Fig:compare} shows that all agents with our algorithm (\ref{update}) do not enter the detection region of radars, while agents 4, 6 and 9 enter the detection region of the radar if all agents choose their directions by the negative gradient method of the original objective function. From Figure \ref{Fig:compare}, our algorithm is more sensitive than the negative gradient method of the original objective function, for that the objective function in our algorithm is a smooth function with better gradient quality.

We also conduct some simulations to test the performance of our algorithm in the dynamic environment. We assume all targets and radar-missiles move randomly at a rate of $10$m/s.
\begin{figure}
    \centering
    \subfloat[]{
        \includegraphics[height=1.5in, width=1.5in]{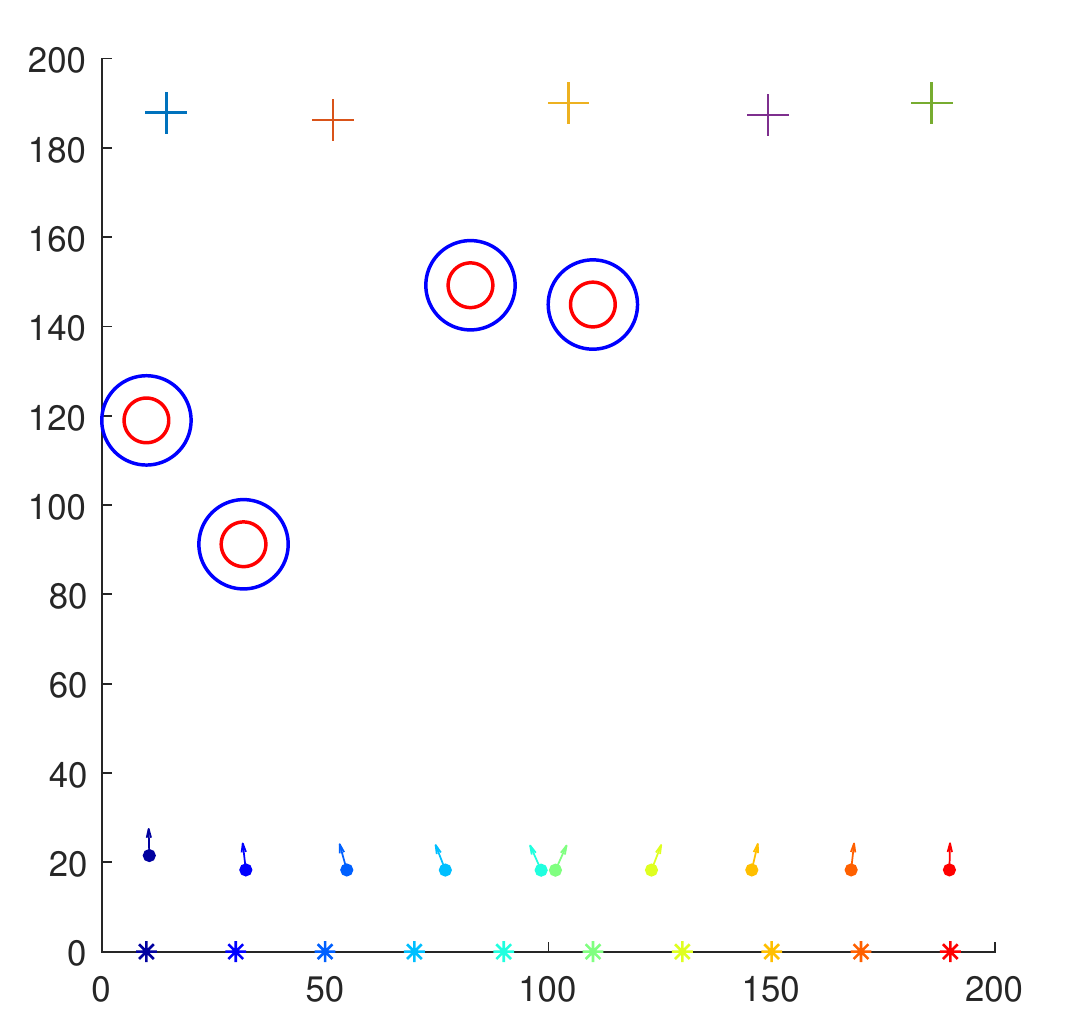}}
    \subfloat[]{
        \includegraphics[height=1.5in, width=1.5in]{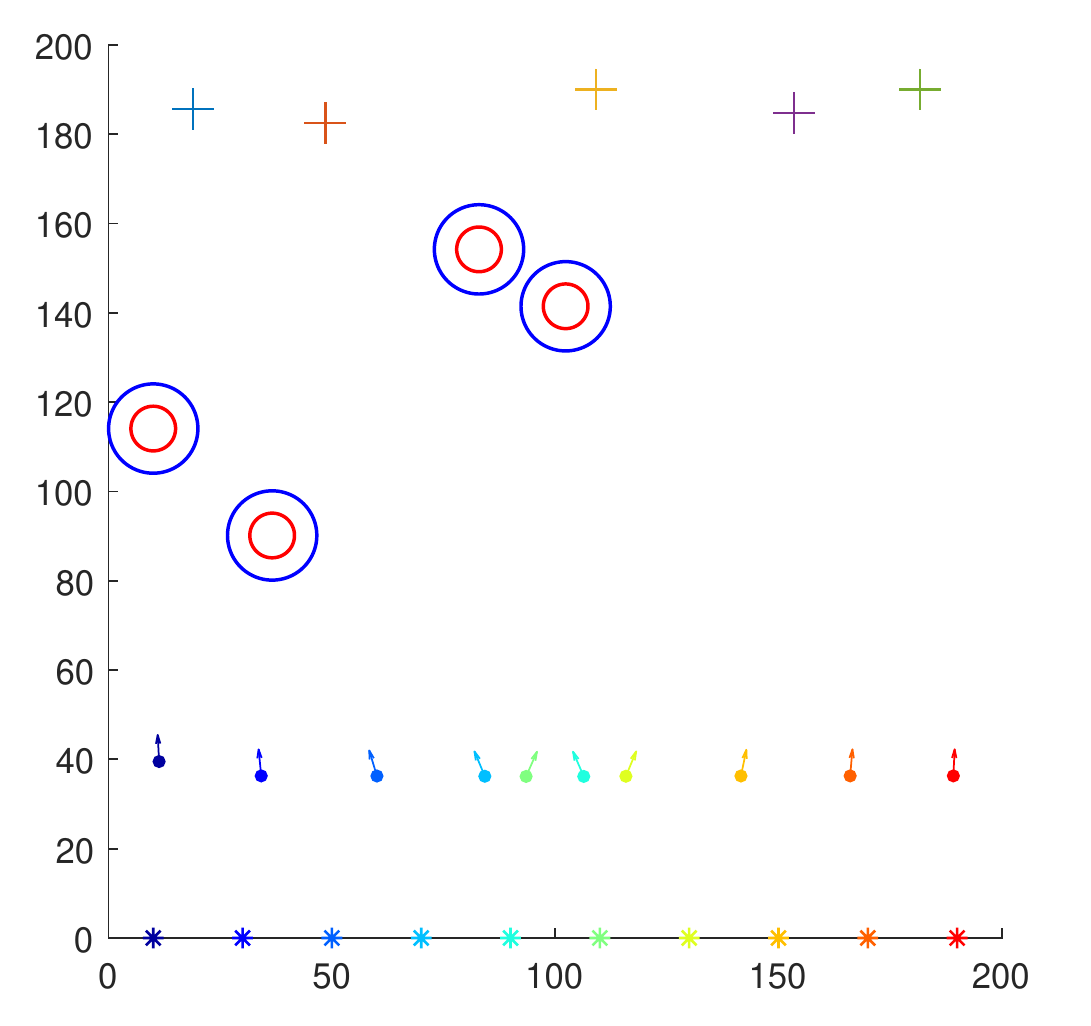}}
    \subfloat[]{
        \includegraphics[height=1.5in, width=1.5in]{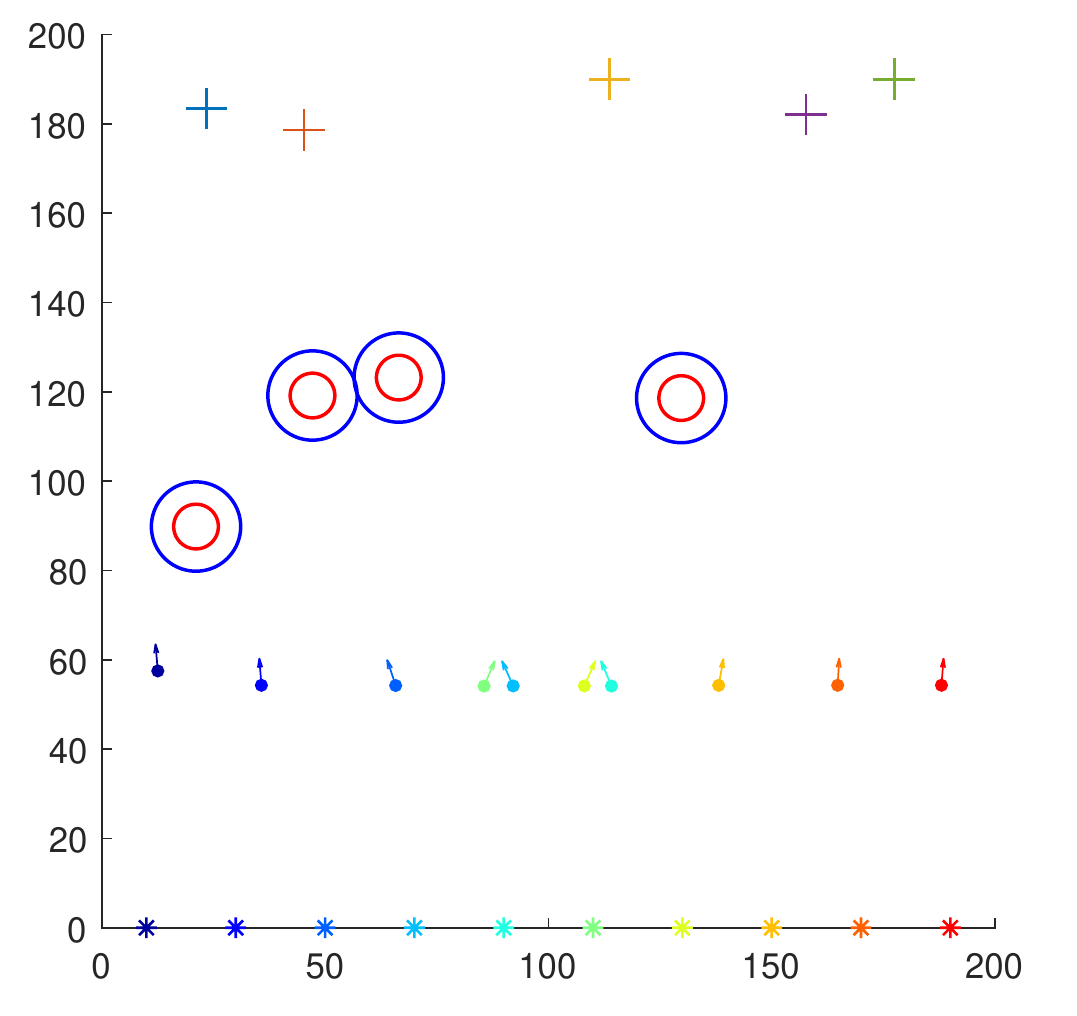}}
    \subfloat[]{
        \includegraphics[height=1.5in, width=1.5in]{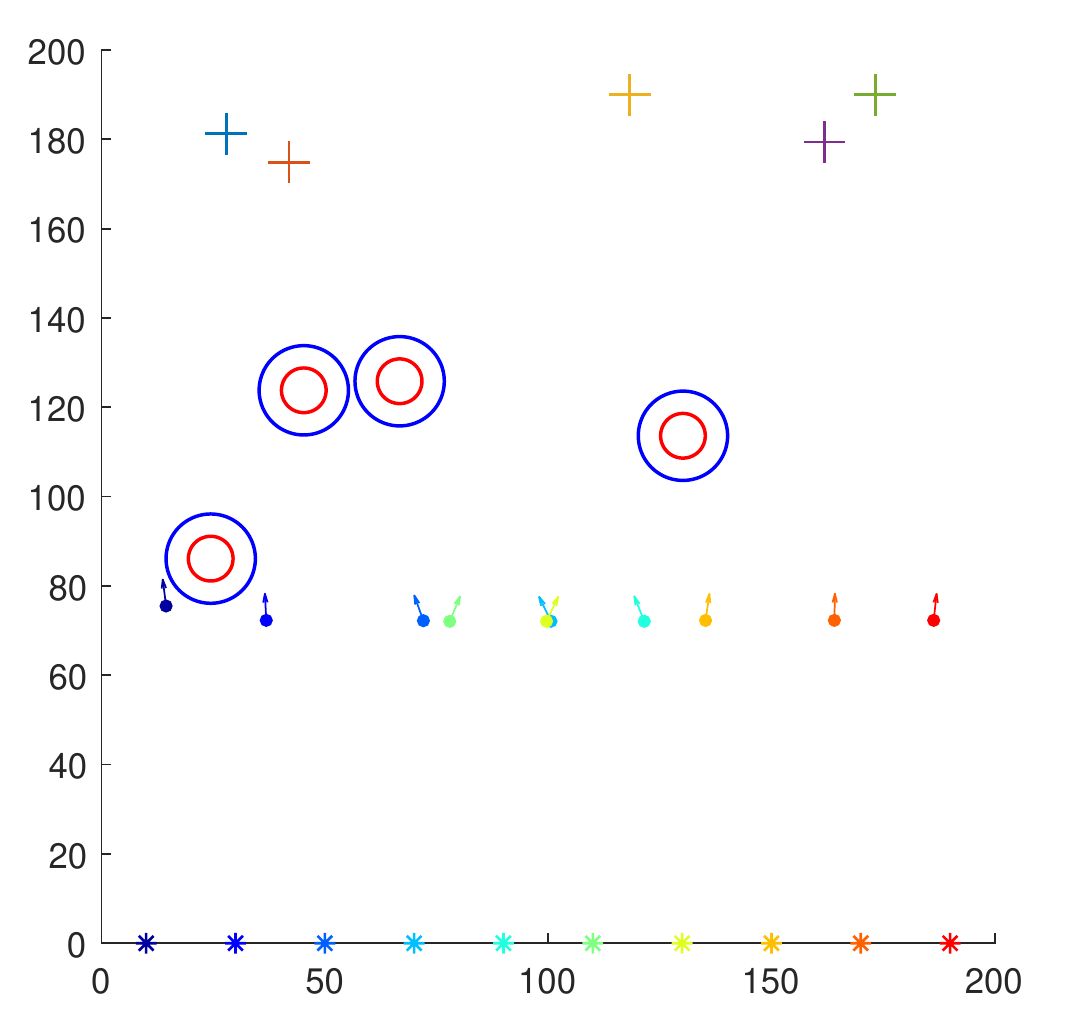}}
    \\
    \subfloat[]{
        \includegraphics[height=1.5in, width=1.5in]{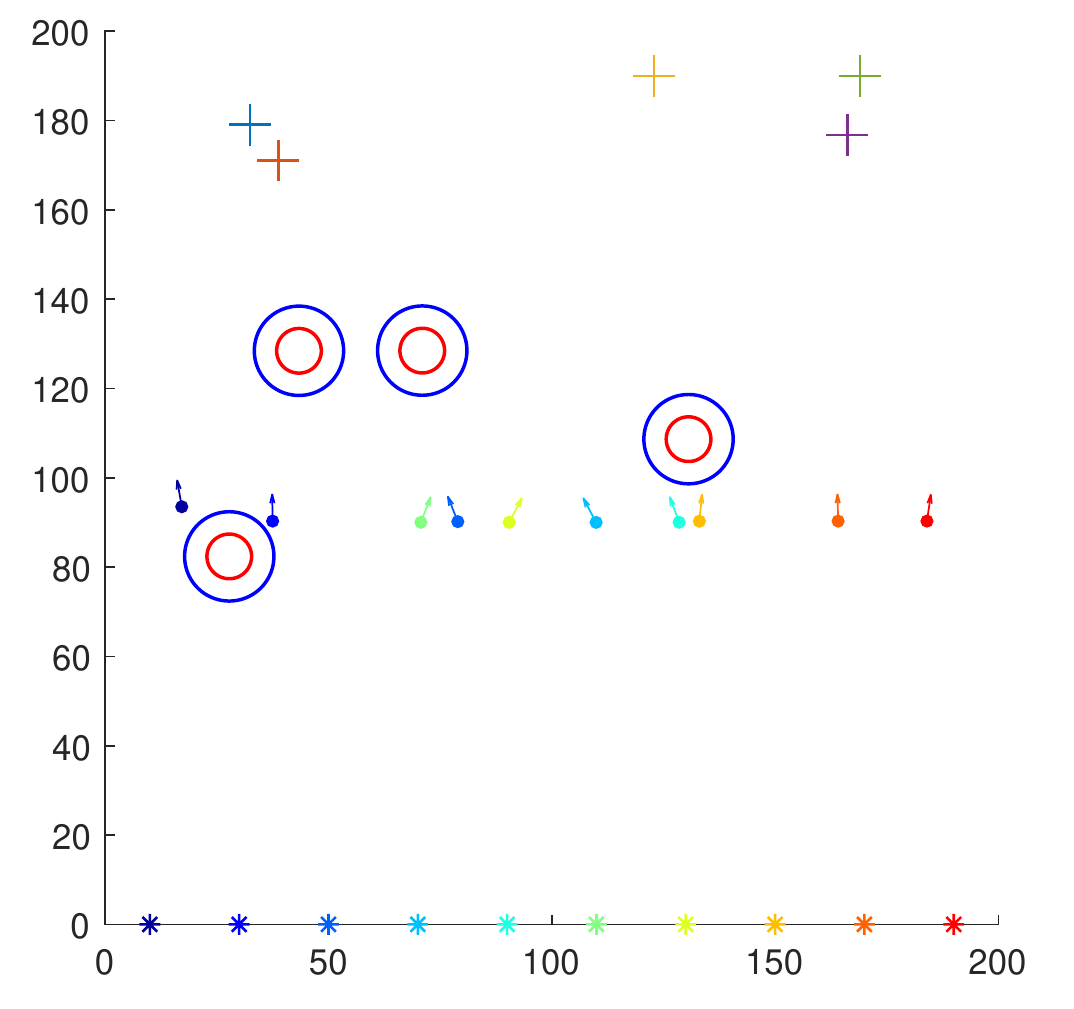}}
    \subfloat[]{
        \includegraphics[height=1.5in, width=1.5in]{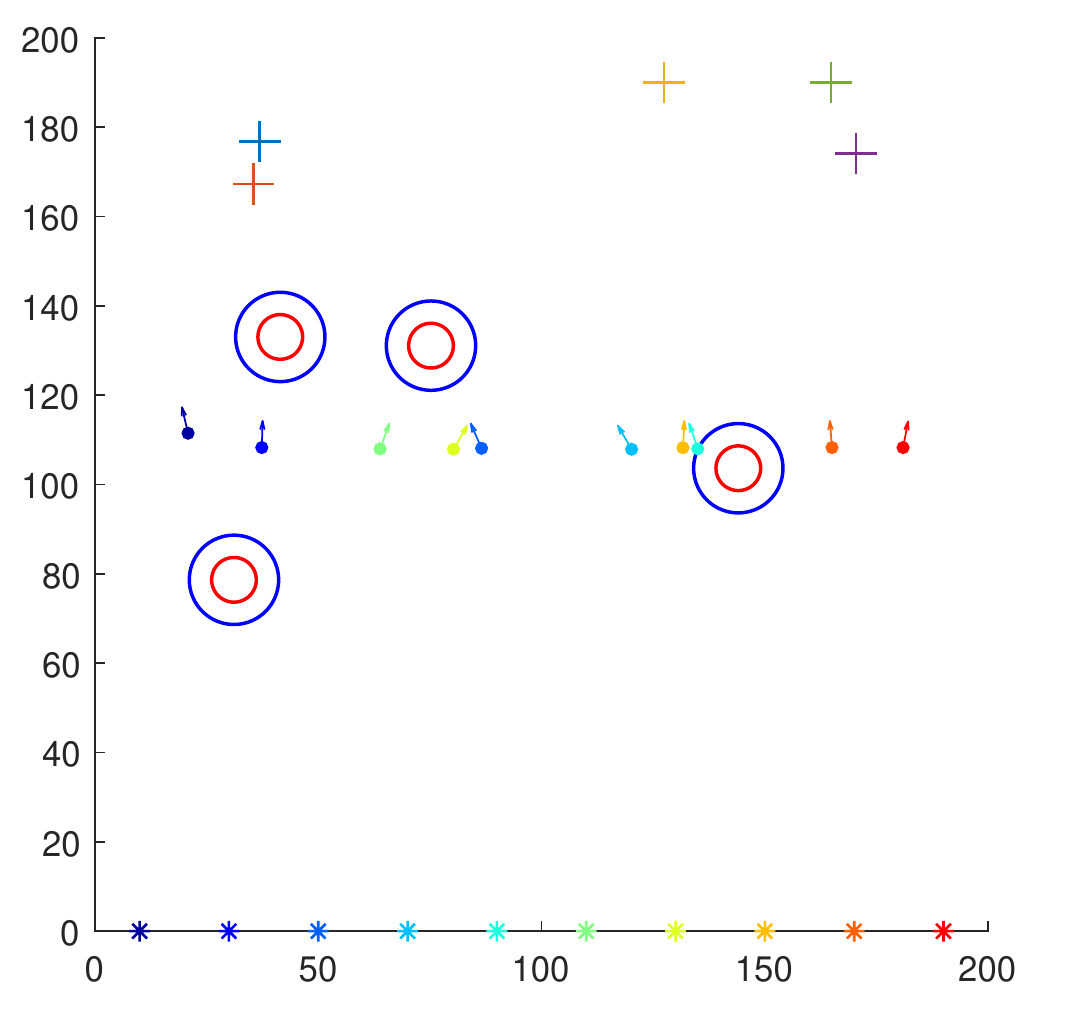}}
    \subfloat[]{
        \includegraphics[height=1.5in, width=1.5in]{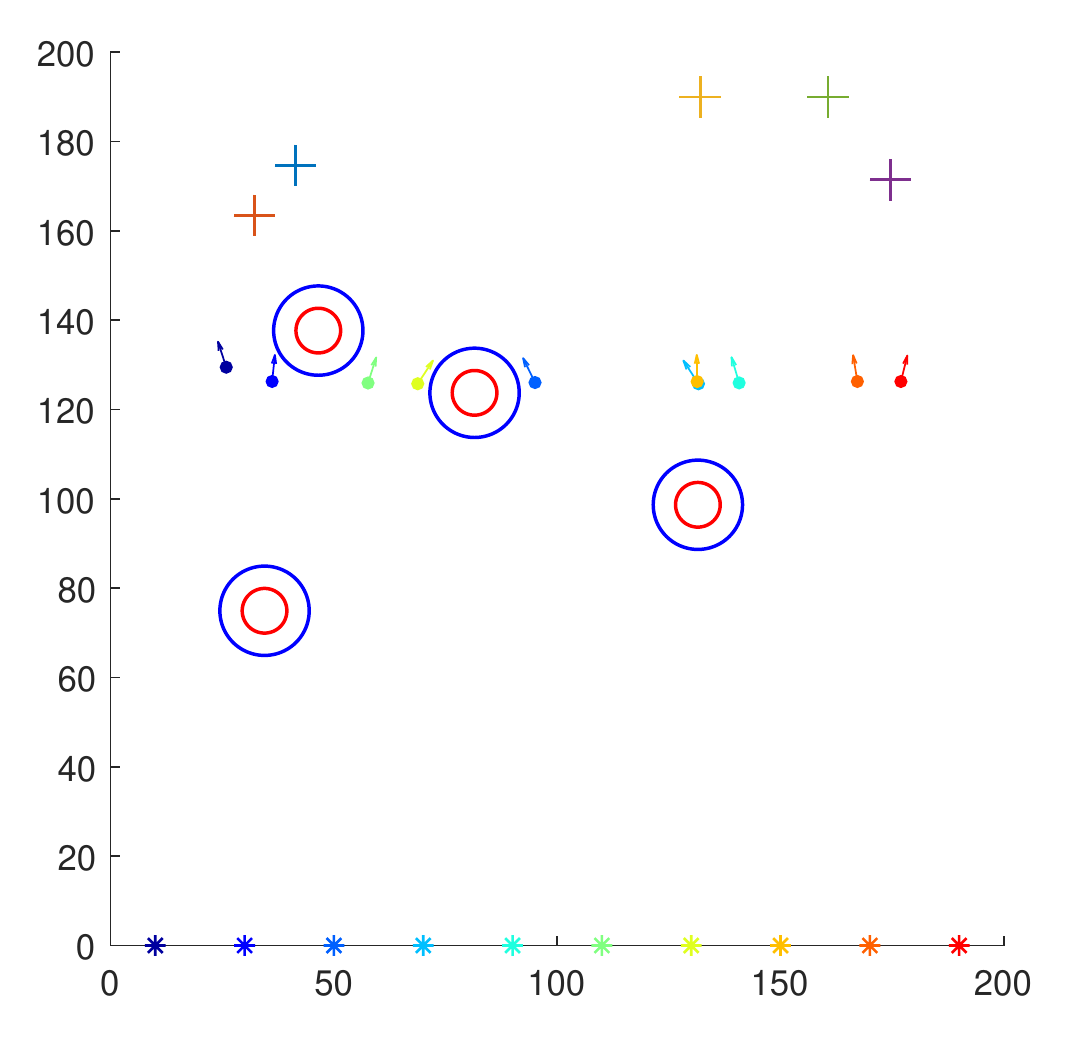}}
    \subfloat[]{
        \includegraphics[height=1.5in, width=1.5in]{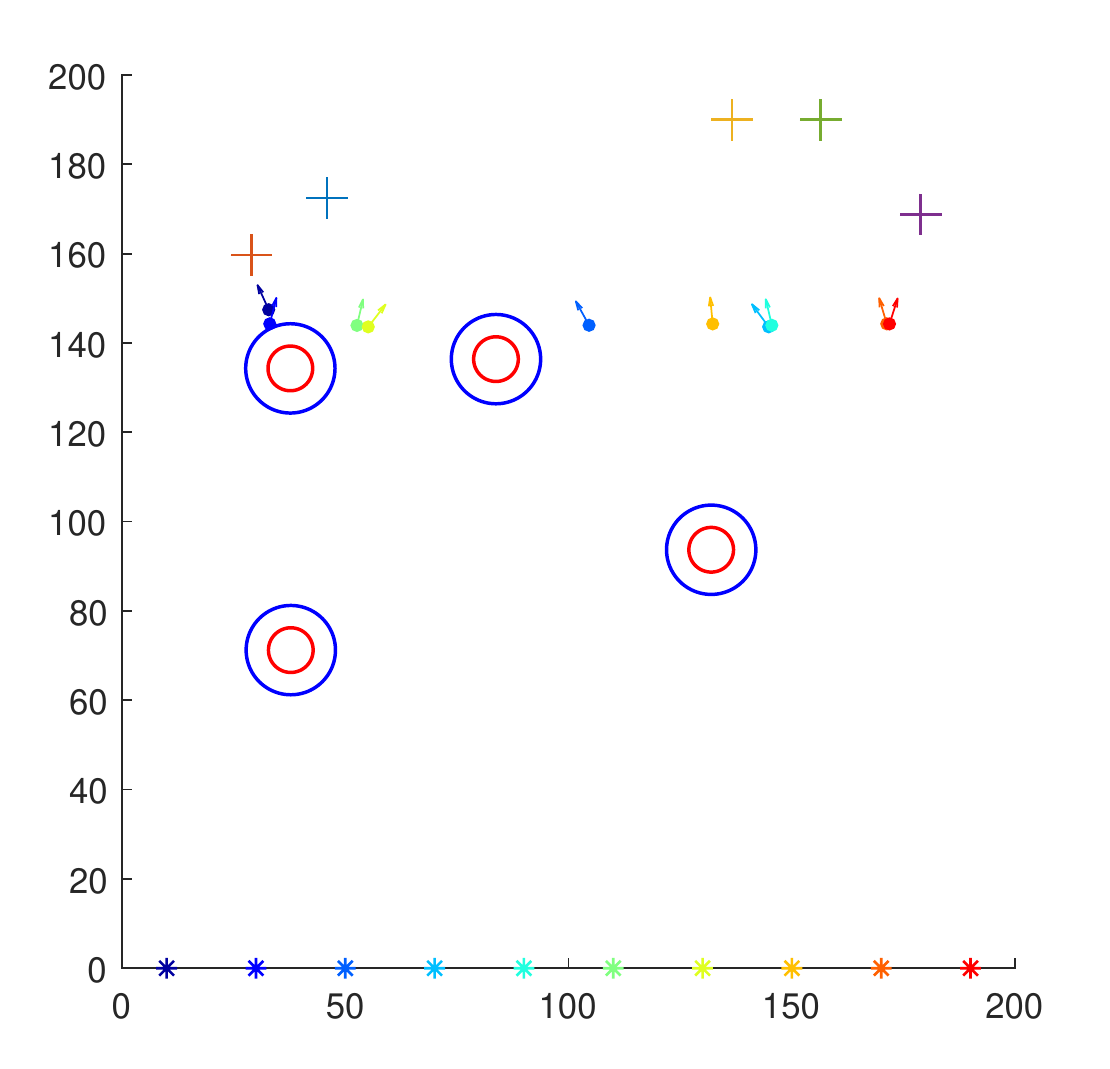}}
    \\
    \subfloat[]{
        \includegraphics[height=1.5in, width=1.5in]{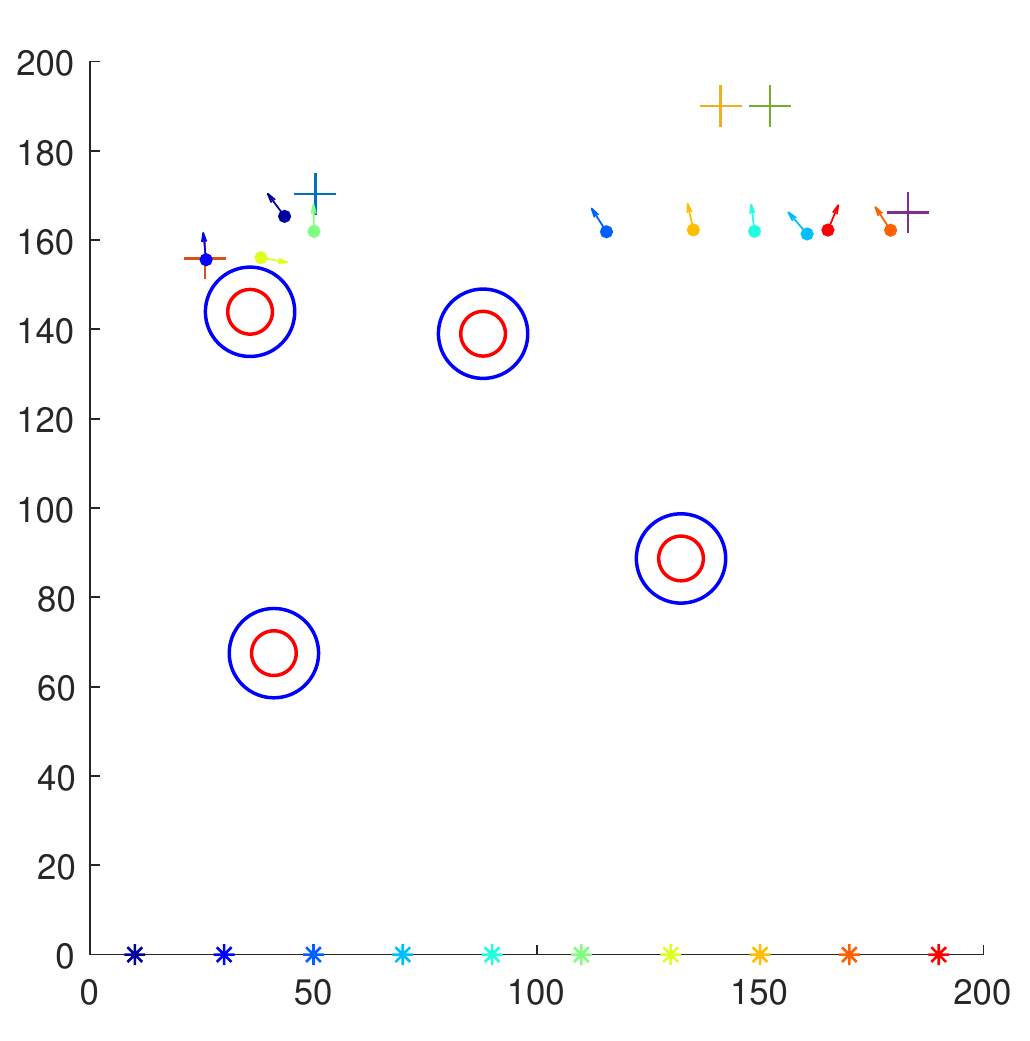}}
    \subfloat[]{
        \includegraphics[height=1.5in, width=1.5in]{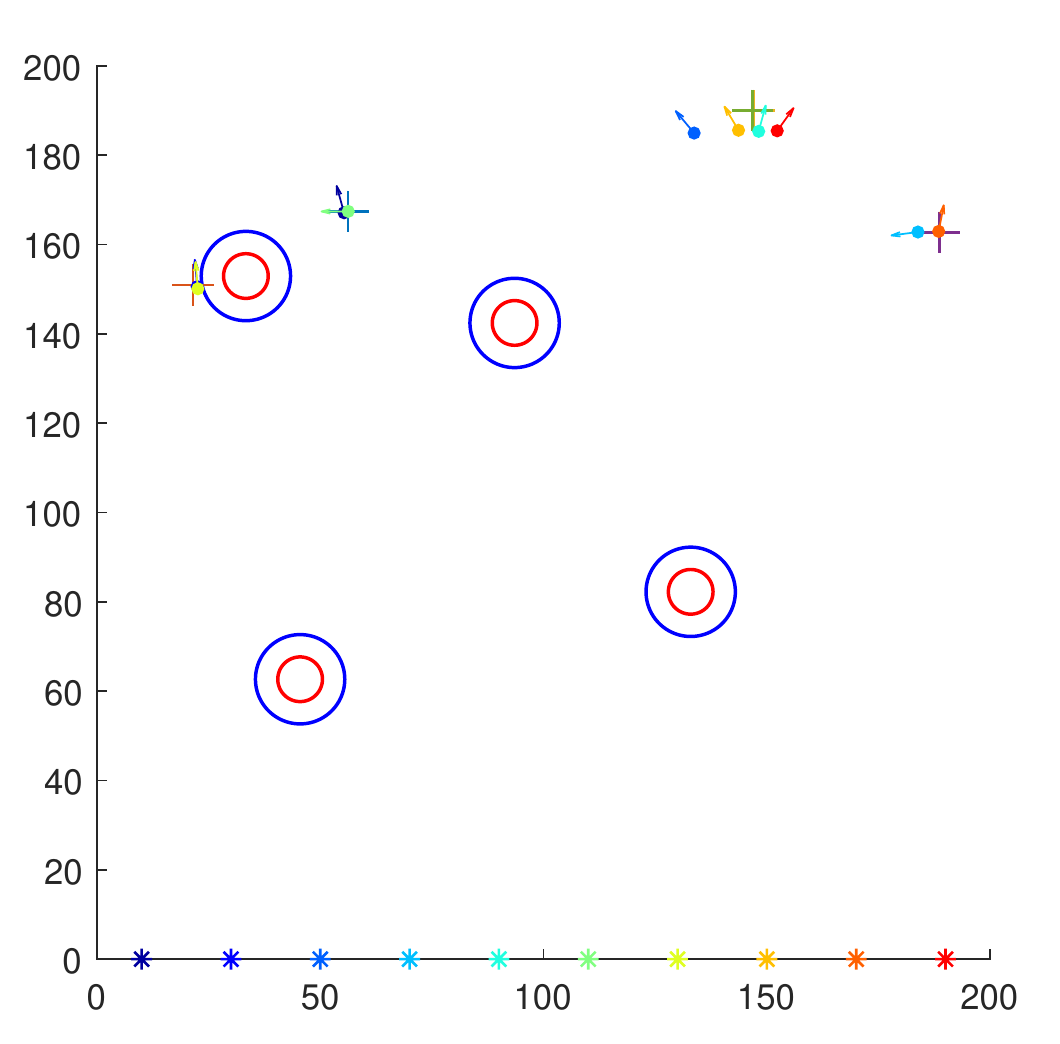}}
\caption{The movement of the agents, targets and radar-missiles in the dynamic environment. (a) At time $t = 100$; (b) At time $t = 200$; (c) At time $t = 300$; (d) At time $t = 400$; (e) At time $t = 500$; (f) At time $t = 600$; (g) At time $t = 700$; (h) At time $t = 800$; (i) At time $t = 900$; (j) At time $t = 1007$, all targets are captured.}
\label{Fig:paths}
\end{figure}
We run our algorithm (\ref{update}) and all agents can reach targets in 1007 time steps. The positions of agents, targets, and radar-missiles at different times are shown in the Figure \ref{Fig:paths}.
To offer more intuitionistic information of our algorithm's performance, we retain the agents that have reached their targets in the Figure \ref{Fig:paths}, and they still track their targets dynamically.
The computational time of each step is shown in the Figure \ref{Fig:time}.
After running in a desktop computer with Intel i7 CPU (2.9 GHZ), the mean computational time of all the steps is 0.0036 s, and the maximum single-step computational time is 0.0125 s, which indicates the low computational requirement of our algorithm.
\begin{figure}[!t]
	\centering
	\includegraphics[height=2.2in, width=2.8in]{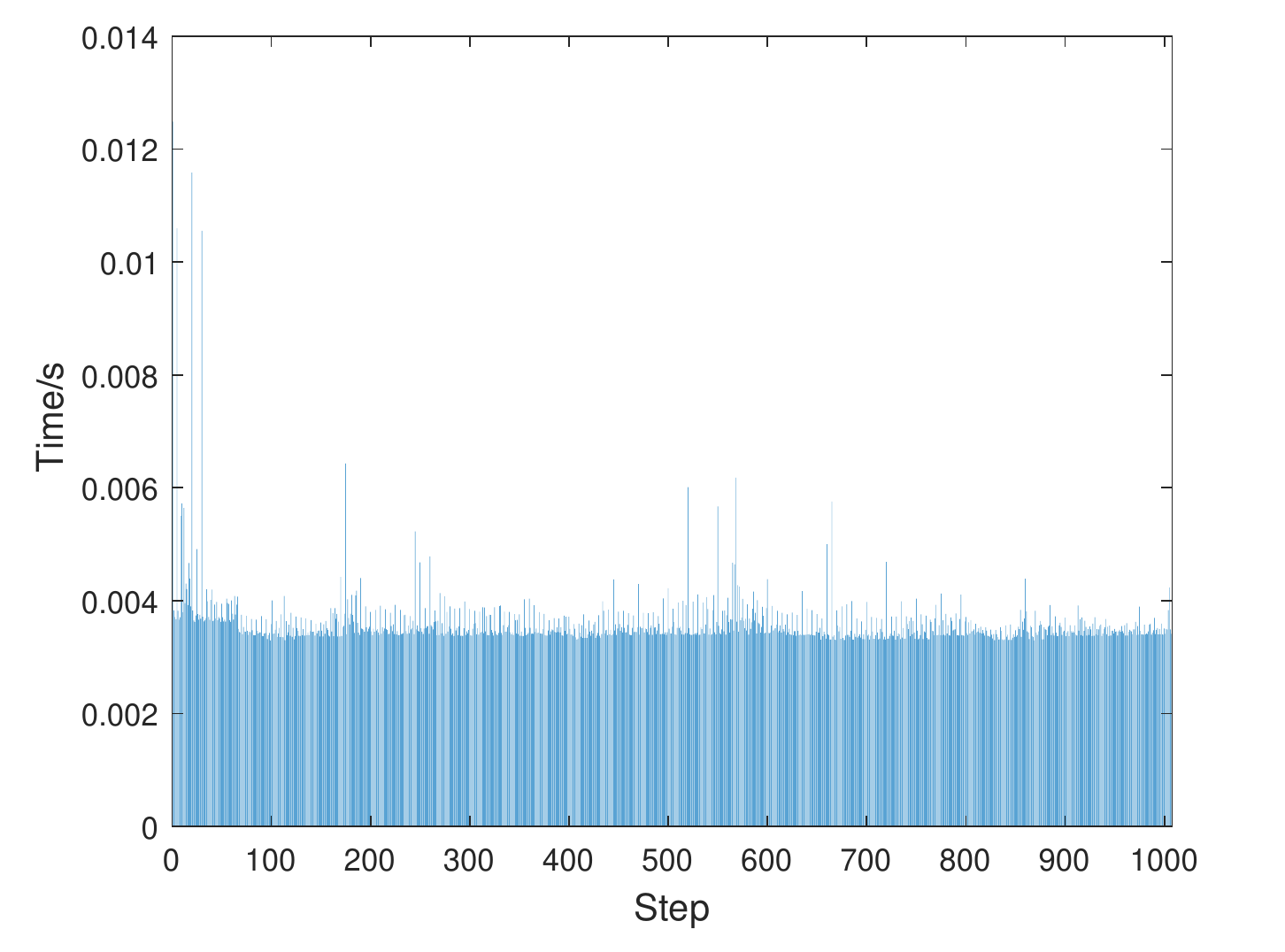}\\
	\caption{The single step computational time of our algorithm. }\label{Fig:time}
\end{figure}

\section{Conclusions}\label{Conclusions}
Real-time path optimization problem is very important to the practical application of USs. However, it is also a challenging problem under fast time-varying and poor communication environments. In this paper, we propose a one-way broadcast communication based algorithm using the weighted Hungarian method, approximation theory with neural networks and modified gradient optimization techniques.
One-way broadcast communication mode is more feasible and applicable under poor communication environments. On this basis, our algorithm can assign targets and find optimal paths for agents in USs in real time. This algorithm can improve the adaptability of agents under the assumption of mobile targets, radars (or sonar), and missiles. From simulations we observe that this algorithm can generate reasonable and optimal paths in a short interval of computational time.

Since the gradient descent method often leads to local optimum, target assignment and path planning are decoupled in our algorithm, and the optimal parameter selection needs further consideration, we need to make more improvements to get closer to the optimal solution, which is hard to achieve due to a lot of uncertainties.
Moreover,  the selection of neural network may need further consideration to enhance the accuracy and convergence rate.
Besides, our simulations are carried out on a computer platform at present, and the following work of embedded control system needs to be assumed and marked out.

\appendix[The weighted Hungarian method]\label{Hungarian method}
First, we review several key definitions\cite{kuhn1955hungarian,konig1931graphs}:
\begin{itemize}

\item[(1)] \textit{Bipartite graph:} A bipartite graph is a special kind of graph. If a graph $\mathcal{G}=(\mathcal{V},\mathcal{E})$, where the vertex set $\mathcal{V}$ can be divided into two disjoint sets $\mathcal{R}$ and $\mathcal{P}$ , and no vertices in the same set are adjacent, the graph $\mathcal{G}$ is a bipartite graph.

\item[(2)] \textit{Matching:} A matching $\mathcal{M}$ of graph $\mathcal{G}$ is a subgraph, in which any two edges have no common vertex.

\item[(3)] \textit{Matched edge} and \textit{matched vertex}: Let  $\mathcal{M}$ be a given matching. All edges in $\mathcal{M}$ are called matched edges, and all endpoints of matched edges are called matched vertices.
	 The other edges and vertices are called unmatched edges and unmatched vertices respectively.

\item[(4)] \textit{Maximum cardinality matching:} Among all matchings in a graph, the matching that contains the most edges is the maximum cardinality matching.

\item[(5)] \textit{Perfect matching:} If the endpoints of a matching $\mathcal{M}$ contain all vertices of a graph, the  matching $\mathcal{M}$ is called a perfect matching. A perfect matching must be a maximum cardinality matching.

\item[(6)] \textit{Alternating path:} An alternating path with respect to a matching $\mathcal{M}$ is a path in which edges alternate between those in $\mathcal{M}$ and those not in $\mathcal{M}$.

\item[(7)] \textit{Alternating path:} An alternating path with respect to a matching $\mathcal{M}$ is a path in which edges alternate between those in $\mathcal{M}$ and those not in $\mathcal{M}$.

\item[(8)] \textit{Augmenting path:} An augmenting path is an alternating path that starts and ends at unmatched vertices, so the number of unmatched edges is one more than the number of matched edges.

\item[(9)] \textit{Vertex cover} and \textit{minimum vertex cover:} A vertex cover of a graph is a set of vertices that includes at least one endpoint of every edge of the graph. A minimum vertex cover is a vertex cover with the smallest possible size.

\item[(10)] \textit{Feasible labeling of vertex:} Let $\mathcal{G}=(\mathcal{V},\mathcal{E},\omega)$ be a weighted graph in which every edge $(i,j)\in\mathcal{E}$ has a weight $\omega(i,j)\in\mathbb{R}$. A vertex labeling $l: \mathcal{V}\rightarrow [0,\infty)$ is a nonnegative function which labels each vertex of $\mathcal{G}$.
	A labeling $l$ is called a feasible labeling if $l(i) + l(j) \leq \omega(i,j)$ for every edge $(i,j)$  in $\mathcal{G}$.

\item[(11)] \textit{Equality subgraph:} Let $\mathcal{G}=(\mathcal{V},\mathcal{E},\omega)$ be a weighted graph and $l$ be a feasible labeling.  An equality subgraph $\mathcal{G}_l=(\mathcal{V},\mathcal{E}_l)$ is a subgraph of $\mathcal{G}$, where
	$$\mathcal{E}_l:=\{(i,j):l(i) + l(j) = \omega(i,j)\}.$$

\end{itemize}	

With the feature of the augmenting path, we can increase the cardinality of matching by constantly looking for an augmenting path. When no augmenting path is found, the current matching is the maximum cardinality matching, which is the principle of the Hungarian method.

The following Kuhn-Munkres theorem is a key theorem of the weighted Hungarian method.
\begin{thm}[Theorem 6 in\cite{kuhn1955hungarian}]\label{KM thm}
If $l$ is feasible and $\mathcal{M}$ is a perfect matching in $\mathcal{G}_l$, then $\mathcal{M}$ is a maximum weighted matching.
\end{thm}

Based on the definitions presented above, we give the  Hungarian method for the assignment of agents and targets. We use the weighted Hungarian method proposed in~\cite{kuhn1955hungarian} to solve the assignment problem for agents, which is described as Algorithm \ref{Hungarian}.
	
\begin{algorithm}
\caption{The weighted Hungarian method}
\label{Hungarian}
\begin{algorithmic}
\STATE {\textbf{Initialization:} In weighted graph $\mathcal{G}=(\mathcal{V},\mathcal{E},\omega)$, $\mathcal{M}$ is a matching with $\mathcal{M} \neq \varnothing$. Let $l$ be an arbitrary feasible labeling of vertex, for example: $l(i)=\min_{j \in \mathcal{P}} \omega(i,j)$ for $i \in \mathcal{R} $, and $l(j)=0$ for $j\in \mathcal{P}$. $\mathcal{G}_l=(\mathcal{V},\mathcal{E}_l)$ is the corresponding equality subgraph.}
\WHILE{$\mathcal{M}$ is not a perfect matching}
\STATE {\textit{Step1: Augment the matching}}
\IF{exist a unmatch vertex $i \in \mathcal{R} $ and we can create an augmenting path}
\STATE {improve the matching by replacing matched edges with unmatched edges in the augmenting path}
\ELSE
\STATE {go to Step2}
\ENDIF
\STATE {\textit{Step2: Improve the label}}
\STATE {\quad\textit{Step2(a): Compute the slack }}
\STATE {\quad Let $\mathcal{S}\subseteq \mathcal{R}$,$\mathcal{T}\subseteq \mathcal{P}$, and $\mathcal{S},\mathcal{T}$ are vertices in the current augmenting path, compute the slack: $\delta_l = \min_{u \in \mathcal{S}, v \notin \mathcal{T}} {l(u)+l(v)-\omega(u,v)} $ }
\STATE {\quad\textit{Step2(b): Improve $l \rightarrow l' $  }}
\STATE {\quad $l'(r) = l(r) - \delta_l$, if $r \in S$}
\STATE {\quad $l'(r) = l(r) + \delta_l$, if $r \in T$}
\STATE {\quad $l'(r) = l(r)$, otherwise}
\ENDWHILE
\end{algorithmic}
\end{algorithm}

\section*{Acknowledgments}
This material is
supported by the Strategic Priority Research Program of Chinese Academy of Sciences under Grant No. XDA27000000, and by the National Natural Science Foundation
of China under grants 12288201, 12071465, and 72192804.


\end{document}